\documentclass[conference]{IEEEtran}
\IEEEoverridecommandlockouts

\usepackage{cite}
\usepackage{amsmath,amssymb,amsfonts}
\usepackage{algorithm}
\usepackage{algorithmic}
\usepackage{graphicx}
\usepackage{textcomp}
\usepackage{xcolor}
\def\BibTeX{{\rm B\kern-.05em{\sc i\kern-.025em b}\kern-.08em
    T\kern-.1667em\lower.7ex\hbox{E}\kern-.125emX}}

\usepackage{amsmath,amssymb,amsthm,bm}
\usepackage{booktabs}
\usepackage{bbm}

\usepackage{tikz}
\usetikzlibrary{arrows.meta,positioning,fit,shapes.geometric}

\newtheorem{definition}{Definition}
\newtheorem{assumption}{Assumption}
\newtheorem{theorem}{Theorem}
\newtheorem{corollary}{Corollary}

\newcommand{\E}{\mathbb{E}}

\newcommand{\LNOQRD}{\textsc{LNO-QRD}}

\begin{document}

\title{
Learning Not to Optimize: Physics-Informed Action-Space Reshaping for Intent-Based \\ Network Control
}

\author{
\IEEEauthorblockN{Zuyuan Zhang}
\IEEEauthorblockA{
\textit{The George Washington University}\\
zuyuan.zhang@gwu.edu
}
\and
\IEEEauthorblockN{Vaneet Aggarwal}
\IEEEauthorblockA{
\textit{Purdue University}\\
vaneet@purdue.edu
}
\and
\IEEEauthorblockN{Tian Lan}
\IEEEauthorblockA{
\textit{The George Washington University}\\
tlan@gwu.edu
}
}

\maketitle

\begin{abstract}
Modern network policy control maps intent to sequential placement-control decisions.
Bellman-style policy optimization primarily asks which action to optimize, while constraints are commonly handled through penalty, barrier, or Lagrangian mechanisms.
We observe that before a value function can certify the best deployment, intermediate signals may already identify many candidates that should be excluded from further optimization.
This motivates a complementary direction: \emph{Learning Not to Optimize}. Before a value function is accurate enough to select the best placement-control decision, intermediate signals may already show that candidates are equivalent under state--intent relabeling (quotienting), lead to a uniformly worse future state (dominance), or violate executable network laws (residual screening). \LNOQRD{} uses these computed or learned signals as a shadow process to reshape the domain on which primal policy optimization is performed, thereby reducing the action space.
We prove lossless quotienting and dominance under explicit equivariance and monotonicity conditions, bound frontier size and ranking cost, and quantify losses from approximate certificates and primal estimates.
Experiments show that \LNOQRD{} reduces small-instance candidates by $75.9\%$ while retaining $90.8\%$ near-oracle coverage and, on large instances, achieves the highest utility and intent satisfaction, the lowest hard-law violation and post-generation latency, and a $73.0\%$ average reduction among candidate-based baselines.
\end{abstract}


\section{Introduction}

Modern programmable networks expose a concrete intent-conditioned placement-control problem.
A translated service-chain or network-slice intent must be realized through joint function placement, flow routing, resource allocation, and admission under time-varying capacity, queues, delay, reliability, isolation, and service dependencies~\cite{clemm2022intent,li2022intent,leivadeas2022survey}.
This problem arises in network function virtualization, service function chaining, network slicing, mobile-edge computing, data-center virtualization, and SDN-based network management~\cite{halpern2015service,mijumbi2015network,bhamare2016survey,foukas2017network,ordonez2017network,sun2022survey,zhang2026lisfc}.
A standard abstraction is a sequential decision problem governed by Bellman optimality, in which a one-step deployment is evaluated together with its future state~\cite{puterman2014markov}.
Existing learning-based controllers mainly improve which deployment is selected through policy learning, graph or attention encoders, candidate ranking, and constrained policy updates~\cite{mao2016resource,mao2019learning,kipf2016semi,velivckovic2017graph,vaswani2017attention,lee2019set,achiam2017constrained,zhang2024distributed,zou2024distributed,zhang2025network}.

This paper starts from a complementary observation: before the controller can certify the best deployment, it may already have enough information to exclude many candidates.
Intermediate signals can reveal state--intent-equivalent deployments~\cite{givan2003equivalence,ravindran2004algebraic,van2020mdp,zhang2026operator,yu2024look}, candidates that leave a uniformly worse residual state under monotone dynamics~\cite{topkis1998supermodularity,smith2002structural,jiang2015approximate,zhanggeometric,zhang2026geometry,zhanghodgeflow}, or violations of executable network laws.
The information required to infer these network-physics-informed exclusion certificates can therefore be substantially weaker than the information required to establish optimality.
Unlike penalty-, barrier-, and Lagrangian-based constrained RL, which enforce constraints within the policy objective or update~\cite{altman2021constrained,achiam2017constrained,tessler2018reward,liu2020ipo,qiao2024br,tang2026nonzero,zhang2026counterfactual}, \LNOQRD{} runs as a shadow process that reshapes the candidate domain before expensive primal ranking.
It converts intermediate primal-learning signals into residual, quotienting, and dominance certificates, thereby excluding candidates before exploration, value estimation, or ranking rather than merely changing their optimization scores.

We instantiate this idea as \emph{Learning Not to Optimize via Quotienting, Residuals, and Dominance} (\LNOQRD).
For an intent-augmented state $x=(s,I)$ and a generated candidate set $\mathcal C_N(x)$, \LNOQRD{} constructs the reduced frontier
$
\mathcal B_\lambda(x)
=
\operatorname{Top}^{\psi_\lambda}_{K_{\rm fr}}\!\left(
\operatorname{ND}_{\lambda,I}\!\left(
\mathcal F^I_\varepsilon(\mathcal C_N(x))/\Gamma_x
\right)
\right)$
Here quotienting by the state-intent stabilizer $\Gamma_x$ merges equivalent placement-control classes; $\mathcal F^I_\varepsilon$ screens candidates by learned or computed violation residuals over resource-state, topology-flow, queue, reliability, and logical-intent laws; and $\operatorname{ND}_{\lambda,I}$ removes dominated quotient classes under fixed-$\lambda$ deployment dynamics. In the end, $\operatorname{Top}^{\psi_\lambda}_{K_{\rm fr}}$ returns a reshaped frontier for the primal policy optimization under desired exploration and execution budget. During training and inference, \LNOQRD{} runs as a shadow process alongside primal policy learning: the primal learner provides residual estimates, structural signatures, and energy scores, while \LNOQRD{} converts these intermediate signals into certificates and feeds the reduced frontier back to the primal learner. We analyze theoretical properties of this closed-loop operation and provide rigorous guarantees.

Theorems~\ref{thm:quotient_optimality} and~\ref{thm:dominance_reduction}, together with Corollary~\ref{cor:exact_reduction}, show that exact quotienting and dominance preserve at least one optimal generated executable deployment under equivariant deployment dynamics and monotone fixed-$\lambda$ values.
Theorem~\ref{thm:frontier_size} bounds the retained set by $\min\{K_{\rm fr},\alpha_{\rm law}(x)\alpha_{\rm nd}(x)\kappa_\Gamma(x)|\mathcal C_N(x)|/|\Gamma_x|\}$, where $\alpha_{\rm law}$, $\alpha_{\rm nd}$, and $\kappa_\Gamma$ denote the law-survival fraction, non-dominated fraction, and maximum stabilizer size; thus ranking cost falls with stronger screening, larger symmetry orbits, and a sparser non-dominated frontier when the saved scoring cost exceeds reduction overhead.
Theorem~\ref{thm:frontier_accounting} and Theorem~\ref{thm:intermediate_signal_error} localize approximation loss to generation, screening, quotienting, dominance, budget truncation, stage slack, and critic error. Consequently, Theorem~\ref{thm:approx_reduced_policy} gives the fixed-$\lambda$ value gap $(\delta+2\epsilon+\tau\log N_{\max})/(1-\gamma)$, making the effects of frontier coverage, energy error, policy temperature, frontier size, and discounting explicit.

We evaluate four placement-control families---service-graph deployment, network-slice realization, service-pipeline deployment, and dynamic drift---with joint placement, routing, allocation, and admission. Small instances use $8$--$10$ nodes, $4$--$5$ tasks, $N=64$, and top-$K=16$; large instances use $124$--$128$ nodes, $10$--$14$ tasks, $N=512$, and top-$K=24$, with traffic, queues, link rates, and reliability varying under drift. All candidate-based methods share the same generator. \LNOQRD{} reduces small-instance candidates by $75.9\%$ while retaining $90.8\%$ near-oracle coverage; on large instances, it attains the highest utility and intent satisfaction, the lowest hard-law violation and post-generation latency, and $73.0\%$ average reduction among candidate-based baselines.

\section{Intent-Conditioned Deployment CMDP and Law Residual Signals}
\label{sec:model}

This section defines only the finite-candidate control interface used by \LNOQRD. We do not model natural-language parsing; an upstream translator maps a request into a structured intent containing service dependencies, traffic demands, hard deployment laws, and optimization preferences \cite{li2022intent,yao2025use,leivadeas2022survey}. Given the current network state and translated intent, the primal controller selects an executable high-utility deployment. The not-to-optimize layer uses the same candidate evaluation to remove generated candidates that are law-violating, dominated, or redundant before expensive primal ranking. Here ``dual'' means this auxiliary reduction role, not the Lagrangian dual of the CMDP.

\subsection{State, Intent, and Deployment Action}

At epoch $t$, the physical network is a directed graph $G_t=(V_t,E_t)$. Let $\mathcal R=\{\mathrm{cpu},\mathrm{mem},\mathrm{sto}\}$ be the resource types. Node and link features are
\begin{align*}
    \xi_v(t)&=\big(C_v^r(t)\big)_{r\in\mathcal R}
    \oplus q_v(t)\oplus \rho_v(t)\oplus\alpha_v(t),\\
    \zeta_e(t)&=\big(B_e(t),d_e(t),\rho_e(t),\eta_e(t)\big),
\end{align*}
where $C_v^r$ is residual node resource, $q_v$ is queue pressure, $B_e$ is residual bandwidth, $d_e$ is link delay, $\rho_v,\rho_e\in(0,1]$ are reliability values, and $\alpha_v,\eta_e$ denote non-capacity node and link attributes. The physical state is
$
    s_t=(G_t,(\xi_v(t))_{v\in V_t},(\zeta_e(t))_{e\in E_t},Q_t)
$
with $Q_t=(q_v(t))_{v\in V_t}$.

A translated intent is $I_t=(H_t,\mathcal D_t,\mathcal C_t,\mathcal O_t)$, where $H_t=(K_t,L_t)$ is a logical service graph, $\mathcal D_t$ contains task and traffic demands, $\mathcal C_t$ contains hard one-step deployment laws and service-level requirements, and $\mathcal O_t$ gives optimization preferences. Each task $k\in K_t$ has demand $r_k=(r_k^r)_{r\in\mathcal R}$, and each logical dependency $\ell=(i,j)\in L_t$ has traffic demand $b_\ell\ge0$. We write the augmented state as
$
    x_t=(s_t,I_t)\in\mathcal X,
$
and omit $t$ when clear.

A deployment action is $a=(M,F,Z,A^{\rm adm})$. Here $M_{kv}\in\{0,1\}$ indicates whether task $k$ is placed on node $v$, $A_k^{\rm adm}\in\{0,1\}$ indicates admission, $F=(f^\ell)_{\ell\in L}$ with $f^\ell\in\mathbb R_+^{|E|}$ is the routing or flow realization, and $z_{kv}^r$ is the resource allocation. Let $\mathcal A_{\rm raw}(x)$ be the syntactically valid action domain induced by $x$. A raw action has valid dimensions and variable types but may still violate deployment laws; only law-satisfying actions are executable.

\subsection{Executable Network Laws}

Executability couples resource, flow, queue, reliability, and intent constraints. Each admitted task must be placed exactly once, and hard dependencies must be admission-consistent:
\begin{align}
    \sum_{v\in V}M_{kv}&=A_k^{\rm adm},
    &&\forall k\in K, \label{eq:placement_law}\\
    A_i^{\rm adm}&=A_j^{\rm adm},
    &&\forall (i,j)\in L_{\rm hard}\subseteq L .
    \label{eq:dependency_admission_law}
\end{align}
For any dependency $\ell=(i,j)$, let $\chi_\ell(A^{\rm adm})\in\{0,1\}$ indicate whether the dependency is active. For hard dependencies, $\chi_\ell=A_i^{\rm adm}=A_j^{\rm adm}$; for optional or degraded dependencies, $\chi_\ell\le A_i^{\rm adm}$ and $\chi_\ell\le A_j^{\rm adm}$. Define $\widetilde b_\ell(a)=b_\ell\chi_\ell(A^{\rm adm})$.

Allocation and capacity laws are
\begin{align}
    z_{kv}^r &\ge M_{kv}r_k^r,
    &&\forall k\in K,\;v\in V,\;r\in\mathcal R,
    \label{eq:resource_lower_law}\\
    z_{kv}^r &\le M_{kv}C_v^r,
    &&\forall k\in K,\;v\in V,\;r\in\mathcal R,
    \label{eq:resource_upper_law}\\
    \sum_{k\in K}z_{kv}^r &\le C_v^r,
    &&\forall v\in V,\;r\in\mathcal R,
    \label{eq:node_capacity_law}\\
    \sum_{\ell\in L}f_e^\ell &\le B_e,
    &&\forall e\in E .
    \label{eq:link_capacity_law}
\end{align}
For flow conservation, let $B\in\mathbb R^{|V|\times|E|}$ be the node-link incidence matrix with $-1$ at the tail and $+1$ at the head, and let $m_k=(M_{kv})_{v\in V}$. For $\ell=(i,j)$,
\begin{align}
    Bf^\ell &= \widetilde b_\ell(a)(m_j-m_i),
    &&\forall \ell\in L, \label{eq:flow_law}\\
    0\le f_e^\ell &\le \widetilde b_\ell(a),
    &&\forall \ell\in L,\;e\in E .
    \label{eq:flow_activation_law}
\end{align}
Thus inactive dependencies carry zero flow, while active dependencies route $b_\ell$ units from the selected location of $i$ to that of $j$.

Intent-level hard laws are written compactly as
\begin{align}
    D(a;x)&\le \bar D(I), \label{eq:model_delay_bound}\\
    R(a;x)&\ge \bar R(I), \label{eq:model_reliability_bound}\\
    h_j(a;x)&\le 0,\qquad j=1,\ldots,m_h,
    \label{eq:model_generic_bound}
\end{align}
where $D$ is end-to-end delay, $R$ is deployment reliability, and $h_j$ represents isolation, affinity, anti-affinity, trust-domain, jurisdiction, stage-ordering, admission, or policy-compliance laws. These functions may be instantiated by any executor-computable model; for example,
\begin{equation}
\begin{split}
    D(a;x)=&\sum_{k,v}M_{kv}D_{kv}^{\rm proc}(z_{kv},q_v)
    +\sum_{\ell,e}\omega_{\ell e}(a)d_e \\
    &+D^{\rm queue}(a;x),
\end{split}
\label{eq:delay_model_revised}
\end{equation}
\begin{equation}
\begin{split}
    \log R(a;x)=&\sum_{k,v}M_{kv}\log\rho_v
    +\sum_{\ell,e}\psi_{\ell e}(a)\log\rho_e,
\end{split}
\label{eq:reliability_model_revised}
\end{equation}
where $\omega_{\ell e}=f_e^\ell/\widetilde b_\ell$ when $\widetilde b_\ell>0$ and $0$ otherwise, and $\psi_{\ell e}$ is either path incidence or a continuous exposure weight. Queue evolution is
\begin{equation}
    Q_{t+1}=\big[Q_t+\Lambda_t-\mu(x_t,a_t)\big]_+,
    \label{eq:queue_model_revised}
\end{equation}
where $\Lambda_t$ is exogenous arrival and $\mu(x_t,a_t)$ is deployment-induced service \cite{tassiulas1990stability,georgiadis2006resource}.

\subsection{Law Residuals and Finite Candidates}

\LNOQRD{} converts executable laws into normalized residual coordinates. Let $\{c_\nu(x,a)\le0\}_{\nu\in\mathcal J_{\le}}$ collect scalar inequality laws and $\{e_\nu(x,a)=0\}_{\nu\in\mathcal J_=}$ collect scalar equality laws, with vector laws treated coordinatewise. For scale $\sigma_\nu>0$,
\begin{equation}
    \bar\phi_\nu(x,a)=
    \begin{cases}
        [c_\nu(x,a)]_+/\sigma_\nu, & \nu\in\mathcal J_{\le},\\[1mm]
        |e_\nu(x,a)|/\sigma_\nu, & \nu\in\mathcal J_= .
    \end{cases}
    \label{eq:normalized_residual_coordinate}
\end{equation}
Let $\bar\phi(x,a)$ collect these coordinates, grouped into placement, dependency, node, flow, link, queue, and intent residuals:
\begin{align}
    \bar\phi(x,a)&=(\bar\phi_\nu(x,a))_{\nu\in\mathcal J_{\le}\cup\mathcal J_=},
    \label{eq:law_residual_vector_model}\\
    \mathsf{Res}_{\rm law}(x,a)&=\|\bar\phi(x,a)\|_2^2 .
    \label{eq:law_penalty_model}
\end{align}
By construction, $\bar\phi(x,a)=0$ iff all encoded hard deployment laws are satisfied. Representative residual components are
\begin{align}
    \bar\phi_{\rm node}(x,a)&=
    \sum_{v,r}\frac{[\sum_k z_{kv}^r-C_v^r]_+}{\sigma_{v,r}^{\rm node}},
    \label{eq:node_residual}\\
    \bar\phi_{\rm link}(x,a)&=
    \sum_e\frac{[\sum_\ell f_e^\ell-B_e]_+}{\sigma_e^{\rm link}},
    \label{eq:link_residual}\\
    \bar\phi_{\rm queue}(x,a)&=
    \sum_v\frac{[\widehat Q_v^+(x,a)-\bar Q_v(I)]_+}{\sigma_v^{\rm queue}},
    \label{eq:queue_residual}\\
    \bar\phi_{\rm intent}(x,a)&=
    \frac{[D(a;x)-\bar D(I)]_+}{\sigma^D}
    +\frac{[\bar R(I)-R(a;x)]_+}{\sigma^R}
    \notag\\
    &\quad +\sum_{j=1}^{m_h}\frac{[h_j(a;x)]_+}{\sigma_j^h},
    \label{eq:intent_residual_model}
\end{align}
where $\widehat Q^+(x,a)=[Q+\Lambda-\mu(x,a)]_+$ and $\bar Q_v(I)$ is the intent-conditioned one-step queue threshold at node $v$. Normalization prevents bandwidth, delay, queue length, or reliability units from dominating only because of scale.

The raw domain $\mathcal A_{\rm raw}(x)$ may be mixed-integer and partially continuous \cite{boyd2004convex,papadimitriou1998combinatorial}. The online interface therefore uses a fixed generated candidate set
\begin{equation}
    \mathcal C_N(x)=
    \{a_i=\mathsf G_i(x):i=1,\ldots,N\}
    \subseteq\mathcal A_{\rm raw}(x),
    \label{eq:candidate_generator_model}
\end{equation}
shared by all compared candidate-based methods. The executable finite-candidate set is
\begin{equation}
    \mathcal E_N(x)=\{a\in\mathcal C_N(x):\bar\phi(x,a)=0\}. 
    \label{eq:executable_candidate_set}
\end{equation}
To keep Bellman maximization defined even when no generated action is executable, define
\begin{equation}
    \widetilde{\mathcal E}_N(x)=
    \begin{cases}
        \mathcal E_N(x), & \mathcal E_N(x)\neq\emptyset,\\
        \{\bot\}, & \mathcal E_N(x)=\emptyset,
    \end{cases}
    \label{eq:effective_executable_set}
\end{equation}
where $\bot$ is an executor-level rejection, degradation, or infeasibility-declaration action with safe transition and fallback reward. If the intent permits controlled rejection, $\bot$ can be implemented through admission control; otherwise it is a safe failure certificate.

\subsection{Finite-Candidate CMDP}

Under a selected action $a_t\in\widetilde{\mathcal E}_N(x_t)$, the augmented state evolves by a Markov kernel $P_X(\cdot\mid x_t,a_t)$. Let $u(x,a)$ be bounded deployment utility and let $g_j(x,a)$ be bounded soft or long-term costs, such as energy, migration, service degradation, or long-term risk. These $g_j$ are distinct from the hard one-step laws encoded in $\bar\phi(x,a)$. A deployed policy satisfies $\pi(a\mid x)=\pi(a\mid s,I)$ with support on $\widetilde{\mathcal E}_N(x)$.

The discounted constrained problem is
\begin{align}
    \max_\pi\quad
    &J_0(\pi)=\mathbb E_\pi
    \left[\sum_{t=0}^{\infty}\gamma^t u(x_t,a_t)\right],
    \label{eq:cmdp_objective_revised}\\
    \text{s.t.}\quad
    &J_j(\pi)=\mathbb E_\pi
    \left[\sum_{t=0}^{\infty}\gamma^t g_j(x_t,a_t)\right]
    \le d_j,
    \quad j=1,\ldots,m,
    \label{eq:cmdp_constraints_revised}
\end{align}
with $\gamma\in(0,1)$. For a fixed multiplier $\lambda\ge0$, define
\begin{equation}
    r_\lambda(x,a)=u(x,a)-\lambda^\top g(x,a).
    \label{eq:lagrangian_reward_revised}
\end{equation}
The fixed-$\lambda$ Bellman quantities over the effective executable interface are
\begin{align}
    V_\lambda^\star(x)&=
    \max_{a\in\widetilde{\mathcal E}_N(x)}Q_\lambda^\star(x,a),
    \label{eq:vstar_revised}\\
    Q_\lambda^\star(x,a)&=
    r_\lambda(x,a)+
    \gamma\mathbb E_{x'\sim P_X(\cdot\mid x,a)}
    [V_\lambda^\star(x')].
    \label{eq:qstar_revised}
\end{align}
Thus the primal optimizer ranks executable candidates by value, while \LNOQRD{} uses additional signals from the same candidate evaluation---law residuals, queue residuals, dependency violations, and intent-preserving equivalence signatures---to decide which raw candidates should not be optimized separately.

\section{\LNOQRD: Learning Not to Optimize via Quotienting, Residuals, and Dominance}
\label{sec:qold}
This section constructs the not-to-optimize frontier used by \LNOQRD. Given the augmented state $x=(s,I)$, the finite-candidate CMDP in Section~\ref{sec:model} asks the primal controller to rank executable deployments. \LNOQRD{} inserts a dual reduction layer before this expensive ranking step: it does not choose the optimal candidate, but removes candidates that are law-violating, state--intent equivalent, dominated, or outside the online budget.

Given $\mathcal C_N(x)=\{a_i\}_{i=1}^N$, the reduction pipeline is
$
\mathcal C_N(x)\to
\mathcal S_\varepsilon(x)\to
\mathcal Q_\varepsilon(x)\to
\mathcal D_\lambda(x)\to
\mathcal B_\lambda(x),
$
where the sets denote, respectively, residual-screened candidates, quotient classes, the non-dominated quotient frontier, and the final budgeted frontier. If no generated candidate is executable, the pipeline returns the fallback action $\bot$ from Section~\ref{sec:model}. The discarded set is $\mathcal N^{\rm no}_\lambda(x)=\mathcal C_N(x)\setminus\mathcal B_\lambda(x)$, and each discarded candidate receives a structural certificate rather than merely a low learned score.

\subsection{Residual Screening}

The first stage removes candidates that violate executable network laws. For each $a\in\mathcal C_N(x)$, \LNOQRD{} evaluates the normalized residual $\bar\phi(x,a)$ and scalar law residual $\mathsf{Res}_{\rm law}(x,a)$ from Section~\ref{sec:model}. For tolerance $\varepsilon\ge0$, let
$
\mathcal F^I_\varepsilon(\mathcal C_N(x))
=
\{a\in\mathcal C_N(x):\mathsf{Res}_{\rm law}(x,a)\le\varepsilon\}.
$
Strict executability uses $\varepsilon=0$, while $\varepsilon>0$ allows numerical margins. To keep the Bellman recursion defined even when all generated deployments violate hard laws, set

\begin{equation}
    \mathcal S_\varepsilon(x)
    =
    \begin{cases}
    \mathcal F^I_\varepsilon(\mathcal C_N(x)),
    &\mathcal F^I_\varepsilon(\mathcal C_N(x))\neq\emptyset,\\[1mm]
    \{\bot\},
    &\mathcal F^I_\varepsilon(\mathcal C_N(x))=\emptyset .
    \end{cases}
    \label{eq:lno_screened_set}
\end{equation}
The fallback $\bot$ is the executor-level rejection, degradation, or infeasibility-declaration action in \eqref{eq:effective_executable_set}, not an infeasible generated deployment. Hence exact residual screening removes only non-executable generated candidates; loss at this stage can only come from finite generation, numerical tolerance, or fallback when no executable candidate is generated.
Operationally, $\varepsilon$ is a screening tolerance rather than permission to deploy an invalid action. Before deployment, the executor recomputes the exact residual; a selected candidate with $\bar\phi(x,a)\neq 0$ is rejected, repaired into an executable action, or mapped to $\bot$. Thus, exact deployment remains supported on $\widetilde{\mathcal E}_N(x)$, while errors introduced by approximate screening are charged to $\Delta_{\rm res}$.

\subsection{State--Intent Quotienting}

Residual screening removes invalid candidates; quotienting removes redundant valid ones. If two deployments differ only by state--intent preserving relabelings of nodes, links, tasks, or dependencies, then the Bellman problem assigns them the same value. Ranking both separately only wastes primal optimization effort.

\begin{definition}[State--intent automorphism]
\label{def:automorphism}
Let $\Gamma$ be a finite relabeling group acting on physical graphs, service graphs, intents, augmented states, and deployment actions. For $x=(s,I)$, define the stabilizer
$
\Gamma_x=\{\gamma\in\Gamma:\gamma x=x\},
$
where each $\gamma\in\Gamma_x$ preserves graph incidence, resource attributes, task demands, traffic demands, and intent constraints; also set $\gamma\bot=\bot$. For $a,b\in\mathcal S_\varepsilon(x)$,
$
a\sim_x b
\Longleftrightarrow
\exists\gamma\in\Gamma_x\text{ such that }b=\gamma a .
$
For finite $U\subseteq\mathcal S_\varepsilon(x)$, $U/\Gamma_x$ denotes the induced quotient set.
\end{definition}

Using Definition~\ref{def:automorphism}, the quotient set is $\mathcal Q_\varepsilon(x)=\mathcal S_\varepsilon(x)/\Gamma_x$. A deterministic map $\operatorname{rep}_x:\mathcal Q_\varepsilon(x)\to\mathcal S_\varepsilon(x)$ selects one representative per class, with $\operatorname{rep}_x(\{\bot\})=\bot$ in the fallback case.

Exact automorphism computation can be expensive on large attributed graphs. We therefore use exact quotienting when available and otherwise use deterministic state--intent signatures over resource, queue, reliability, domain, task-type, traffic, and path features. Exact quotienting is lossless; signature quotienting is conservative if it only splits exact orbits and approximate if it merges non-equivalent actions. The latter loss is accounted for in Section~\ref{sec:guarantees}.

\begin{assumption}[Equivariant deployment interface]
\label{ass:equiv_deployment}
For every $\gamma\in\Gamma$, augmented state $x$, action $a$, measurable $\mathcal B\subseteq\mathcal X$, and multiplier $\lambda\ge0$, the executable interface, fixed-$\lambda$ reward, and transition kernel satisfy
$
\widetilde{\mathcal E}_N(\gamma x)=\gamma\widetilde{\mathcal E}_N(x),\quad
r_\lambda(\gamma x,\gamma a)=r_\lambda(x,a),\quad
P_X(\mathcal B\mid\gamma x,\gamma a)=P_X(\gamma^{-1}\mathcal B\mid x,a).
$
\end{assumption}

\begin{theorem}[Lossless quotienting]
\label{thm:quotient_optimality}
Under Assumption~\ref{ass:equiv_deployment}, for every $\lambda\ge 0$,
\begin{equation}
    Q_\lambda^\star(\gamma x,\gamma a)
    =
    Q_\lambda^\star(x,a),
    \qquad
    \forall \gamma\in\Gamma .
    \label{eq:q_equivariance}
\end{equation}
Consequently, for fixed $x$ and any $\gamma\in\Gamma_x$,
$
    Q_\lambda^\star(x,\gamma a)=Q_\lambda^\star(x,a).
$
Therefore, keeping one representative per exact orbit in $\mathcal S_\varepsilon(x)/\Gamma_x$ preserves the best fixed-$\lambda$ generated action value inside $\mathcal S_\varepsilon(x)$.
\end{theorem}

\begin{proof}[Proof sketch]
Let $T_\lambda$ be the fixed-$\lambda$ Bellman optimality operator and define $Q^\gamma(x,a)=Q(\gamma x,\gamma a)$. Assumption~\ref{ass:equiv_deployment} gives the operator equivariance $T_\lambda(Q^\gamma)=(T_\lambda Q)^\gamma$. Since the discounted Bellman operator has a unique fixed point, $(Q_\lambda^\star)^\gamma=Q_\lambda^\star$, which yields \eqref{eq:q_equivariance}. If $\gamma\in\Gamma_x$, then $\gamma x=x$, so all actions in the same orbit have identical optimal action value at $x$.
\end{proof}

Theorem~\ref{thm:quotient_optimality} applies only to exact state--intent symmetries. Exact quotienting identifies actions that are value-equivalent under Assumption~\ref{ass:equiv_deployment}, whereas signature quotienting is approximate and its induced error must be quantified.

\subsection{Dominance Pruning}

Quotienting removes actions that are equivalent. Dominance pruning removes actions that are different but provably no better. The key point is dynamic: a candidate can be worse not only because it has smaller immediate Lagrangian reward, but also because it leaves a weaker residual network state for future decisions.

Let $\succeq_{\mathcal X}$ be a preorder over augmented states: $x^+\succeq_{\mathcal X}x^-$ means that $x^+$ has no smaller residual capacity, no larger queue pressure, no worse reliability, and the same or no harder intent class than $x^-$. Non-monotone coordinates for a specific deployment problem are excluded. For kernels at the same state, write $P_X(\cdot\mid x,a)\succeq_{\rm st}P_X(\cdot\mid x,b)$ if $\mathbb E[\psi(x')\mid x,a]\ge\mathbb E[\psi(x')\mid x,b]$ for every bounded increasing $\psi$ under $\succeq_{\mathcal X}$.

\begin{definition}[Deployment dominance]
\label{def:dominance}
For fixed $x$ and $\lambda$, action $a$ dominates $b$, written $a\succeq_{x,\lambda}b$, if
$
r_\lambda(x,a)\ge r_\lambda(x,b)
$
and
$
P_X(\cdot\mid x,a)\succeq_{\rm st}P_X(\cdot\mid x,b).
$
The strict relation $a\succ_{x,\lambda}b$ holds when at least one comparison is strict.
\end{definition}

\begin{assumption}[Monotone fixed-$\lambda$ value]
\label{ass:monotone_value}
For the fixed multiplier $\lambda$, the optimal value is increasing under $\succeq_{\mathcal X}$:
$
    x^+\succeq_{\mathcal X}x^-
    \quad\Longrightarrow\quad
    V_\lambda^\star(x^+)\ge V_\lambda^\star(x^-).
$
\end{assumption}
This is the standard monotone dynamic-programming condition: additional residual capacity, lower queue pressure, and better reliability cannot reduce future deployment opportunities \cite{puterman2014markov,smith2002structural,topkis1998supermodularity,jiang2015approximate}.

\begin{theorem}[Lossless dominance pruning]
\label{thm:dominance_reduction}
Under Assumption~\ref{ass:monotone_value}, if $a\succeq_{x,\lambda}b$, then
$
    Q_\lambda^\star(x,a)\ge Q_\lambda^\star(x,b).
$
Therefore, a strictly dominated action cannot be the unique optimal deployment at $x$.
\end{theorem}

\begin{proof}[Proof sketch]
By the Bellman equation,
$
Q_\lambda^\star(x,a)-Q_\lambda^\star(x,b)
=
r_\lambda(x,a)-r_\lambda(x,b)
+\gamma\Big(
\E[V_\lambda^\star(x')\mid x,a]
-\E[V_\lambda^\star(x')\mid x,b]
\Big).
$
The reward difference is nonnegative by Definition~\ref{def:dominance}. The expectation difference is nonnegative because Assumption~\ref{ass:monotone_value} makes $V_\lambda^\star$ increasing and the transition under $a$ stochastically dominates that under $b$. Hence $Q_\lambda^\star(x,a)\ge Q_\lambda^\star(x,b)$.
\end{proof}

Exact stochastic dominance is rarely available online. We therefore use a sufficient diagnostic certificate based on the normalized badness vector
$\mathbf b(x,a)=
\left(
\begin{smallmatrix}
\widetilde D(a;x),\, C^{\rm cost}(a;x),\, 1-R(a;x),\, \mathsf{Res}_{\rm law}(x,a)\\
-S^{\rm node}(x,a),\, -S^{\rm link}(x,a),\, L^{\rm path}(x,a),\, I^{\rm logic}(x,a)
\end{smallmatrix}
\right).$
All coordinates are normalized. Here $\widetilde D$ is normalized delay, $C^{\rm cost}$ is deployment cost, $S^{\rm node}$ and $S^{\rm link}$ are residual node and link slacks, $L^{\rm path}$ is routing burden, and $I^{\rm logic}$ is the policy-logic residual; $1-R$ and $\mathsf{Res}_{\rm law}$ measure unreliability and total law residual. Candidate $b$ is pruned if some $a$ satisfies $\mathbf b(x,a)\le \mathbf b(x,b)$ coordinate-wise and is strictly smaller in at least one coordinate. This certificate is exact only when the badness-vector order implies Definition~\ref{def:dominance}; otherwise its loss is charged to the dominance term in Section~\ref{sec:guarantees}.

For quotient classes, define
$
    c_a\succeq^q_{x,\lambda}c_b
    \quad\Longleftrightarrow\quad
    \operatorname{rep}_x(c_a)
    \succeq_{x,\lambda}
    \operatorname{rep}_x(c_b).
$
The non-dominated quotient frontier is
$
    \mathcal D_\lambda(x)
    =
    \operatorname{ND}_{\lambda,I}(\mathcal Q_\varepsilon(x))
    =
    \left\{
    c\in\mathcal Q_\varepsilon(x):
    \nexists c'\in\mathcal Q_\varepsilon(x)
    \text{ such that } c'\succ^q_{x,\lambda}c
    \right\}.
$

\begin{corollary}[Exact structural reduction]
\label{cor:exact_reduction}
Suppose residual screening is exact and returns at least one executable generated candidate, quotienting is exact, and dominance is exact on the finite generated candidate set. Under Assumptions~\ref{ass:equiv_deployment} and~\ref{ass:monotone_value}, reducing $\mathcal C_N(x)$ to
$
    \{\operatorname{rep}_x(c):
    c\in
    \operatorname{ND}_{\lambda,I}(\mathcal S_0(x)/\Gamma_x)\}
$
preserves at least one optimal generated executable deployment.
\end{corollary}

\begin{proof}[Proof sketch]
Exact residual screening removes only non-executable generated candidates. Theorem~\ref{thm:quotient_optimality} preserves the best value within each exact orbit, and Theorem~\ref{thm:dominance_reduction} rules out strictly dominated representatives as unique optima. Finiteness then guarantees that at least one best generated executable deployment remains.
\end{proof}

\subsection{Budgeted Composite-Priority Frontier}

The final stage enforces the online inference budget. It is not a feasibility certificate and is lossless only when the retained budget includes a near-optimal representative. Let $S(x,a)=S^{\rm node}(x,a)+S^{\rm link}(x,a)$, let $\widehat U_{\rm ub}(x,a)$ be the learned upper utility score, and take all weights $w_\cdot\ge0$. For finite representative set $U$, let $\operatorname{Top}^{\psi_\lambda}_{K_{\rm fr}}(U)$ return the $K_{\rm fr}$ candidates with smallest diagnostic priority, using deterministic tie-breaking:
$
\psi_\lambda(x,a)=
w_\phi\|\bar\phi(x,a)\|_1
+w_d\widetilde D(a;x)
+w_c C^{\rm cost}(a;x)
+w_r(1-R(a;x))
-w_s S(x,a)
+w_p L^{\rm path}(a;x)
-w_u\widehat U_{\rm ub}(x,a).
$
The budgeted frontier is
$
\mathcal B_\lambda(x)=
\operatorname{Top}^{\psi_\lambda}_{K}
(\{\operatorname{rep}_x(c):c\in\mathcal D_\lambda(x)\}).
$
If top-$K$ removes a high-value candidate, the loss is accounted for by $\Delta_K$ in Section~\ref{sec:guarantees}.

\subsection{Frontier Size and Ranking Speedup}

The previous subsections define what \LNOQRD{} removes. We now state the resulting reduction in expensive primal ranking calls. This is the formal speedup result for the dual not-to-optimize layer: the primal scorer is applied only to the retained frontier, while residuals, quotienting, and dominance use cheaper structural information produced during candidate evaluation.

Assume that $\mathcal F^I_\varepsilon(\mathcal C_N(x))$ is nonempty, finite, and invariant under $\Gamma_x$. Let
$
\operatorname{Stab}_x(a)=\{\gamma\in\Gamma_x:\gamma a=a\},
$
$
\kappa_\Gamma(x)=\max_{a\in\mathcal F^I_\varepsilon(\mathcal C_N(x))}|\operatorname{Stab}_x(a)|,
$
and
$
\alpha_{\rm law}(x)=
|\mathcal F^I_\varepsilon(\mathcal C_N(x))|/|\mathcal C_N(x)|,\quad
\alpha_{\rm nd}(x)=|\mathcal D_\lambda(x)|/|\mathcal Q_\varepsilon(x)|.
$
Set $\kappa_\lambda(x)=\alpha_{\rm law}(x)\alpha_{\rm nd}(x)\kappa_\Gamma(x)$.

\begin{theorem}[Frontier complexity and ranking speedup]
\label{thm:frontier_size}
For every state--intent pair $x=(s,I)$ with nonempty $\mathcal F^I_\varepsilon(\mathcal C_N(x))$,
\begin{equation}
    |\mathcal B_\lambda(x)|
    \le
    \min\left\{
    K,\;
    \frac{\kappa_\lambda(x)}{|\Gamma_x|}
    |\mathcal C_N(x)|
    \right\}.
    \label{eq:frontier_size_bound}
\end{equation}
Suppose evaluating the primal ranking model on one candidate costs $c_E$, and let $T_{\rm red}(x)$ be the online cost of residual evaluation, quotienting, dominance testing, and top-$K$ selection after candidates have been generated. Raw primal ranking costs
$
    T_{\rm raw}(x)=c_E|\mathcal C_N(x)|.
$
Ranking after \LNOQRD{} costs at most
$
T_{\rm LNO}(x)
\le
T_{\rm red}(x)
+c_E
\min\left\{
K_{\rm fr},\;
\frac{\kappa_\lambda(x)}{|\Gamma_x|}
|\mathcal C_N(x)|
\right\}.
$
Thus, \LNOQRD{} reduces online ranking cost whenever
\begin{equation}
    T_{\rm red}(x)
    <
    c_E
    \left(
    |\mathcal C_N(x)|
    -
    \min\left\{
    K_{\rm fr},\;
    \frac{\kappa_\lambda(x)}{|\Gamma_x|}
    |\mathcal C_N(x)|
    \right\}
    \right).
    \label{eq:lno_speedup_condition}
\end{equation}
If $\mathcal F^I_\varepsilon(\mathcal C_N(x))=\emptyset$, then $\mathcal B_\lambda(x)=\{\bot\}$ and no expensive primal ranking over generated deployments is required.
\end{theorem}

\begin{proof}[Proof sketch]
For any feasible generated candidate, orbit--stabilizer gives
$
|[a]_x|=|\Gamma_x|/|\operatorname{Stab}_x(a)|
\ge |\Gamma_x|/\kappa_\Gamma(x).
$
Since $\mathcal F^I_\varepsilon(\mathcal C_N(x))$ is $\Gamma_x$-invariant,
$
|\mathcal Q_\varepsilon(x)|
\le
\kappa_\Gamma(x)|\mathcal F^I_\varepsilon(\mathcal C_N(x))|/|\Gamma_x|.
$
Substituting the retained law fraction $\alpha_{\rm law}$ and non-dominated fraction $\alpha_{\rm nd}$, then applying the top-$K$ cap, gives \eqref{eq:frontier_size_bound}. The cost statement follows because the expensive primal scorer is evaluated only on $\mathcal B_\lambda(x)$.
\end{proof}

\begin{corollary}[Predictive compression bound]
\label{cor:predictive_bound}
Suppose that, over a class of states, residual screening keeps at most a fraction $p_{\rm law}$ of generated candidates, non-dominated filtering keeps at most a fraction $p_{\rm nd}$ of quotient classes, and every action stabilizer has size at most $s$. Then
\begin{equation}
    |\mathcal B_\lambda(x)|
    \le
    \min\left\{
    K_{\rm fr},\;
    \frac{p_{\rm law}p_{\rm nd}s}{|\Gamma_x|}
    |\mathcal C_N(x)|
    \right\}.
    \label{eq:predictive_size_bound}
\end{equation}
\end{corollary}

Corollary~\ref{cor:predictive_bound} characterizes predictive compression through the residual feasibility rate, quotient redundancy rate, non-dominated frontier fraction, final ranking budget, and measured reduction overhead.

\subsection{Intent-Conditioned Deployment Energy}

After \LNOQRD{} constructs $\mathcal B_\lambda(x)$, the primal learner ranks only this frontier using the deployment energy
$
\mathcal E_\theta(x,a)
=
\widehat C_\theta(x,a)
+\lambda^\top\widehat G_\theta(x,a)
+\eta_{\rm law}\mathsf{Res}_{\rm law}(x,a),
$
where $\widehat C_\theta$ is a learned negative utility-to-go proxy, $\widehat G_\theta$ estimates discounted soft CMDP cost, and the hard-law residual remains an executor-facing signal. The reduced Boltzmann policy is
\begin{equation}
    \pi_\theta(a\mid x)
    =
    \frac{\exp(-\mathcal E_\theta(x,a)/\tau)}
    {\sum_{b\in\mathcal B_\lambda(x)}
    \exp(-\mathcal E_\theta(x,b)/\tau)},
    \qquad
    a\in\mathcal B_\lambda(x).
    \label{eq:energy_policy_revised}
\end{equation}
At evaluation, the controller may take $a_t\in\arg\min_{a\in\mathcal B_\lambda(x)}\mathcal E_\theta(x,a)$.
Here $\tau>0$ is the policy temperature and $\eta_{\rm law}\ge0$ controls residual-aware energy shaping. Consequently, the energy model focuses on value differences within the reduced frontier rather than repeatedly distinguishing candidates already excluded by explicit structural certificates.

The energy model uses permutation-invariant or equivariant encoders for candidate sets, physical graphs, and service graphs \cite{zaheer2017deep,lee2019set,vaswani2017attention,kipf2016semi,velivckovic2017graph}. 
It is trained to respect the network symmetry group through 
$\mathcal E_\theta(\gamma x,\gamma a)=\mathcal E_\theta(x,a)$ for $\gamma\in\Gamma_x$. 
The learning objective augments PPO with value fitting and three structured regularizers: 
a soft-cost regression term $\mathcal L_g$, a certified-order ranking term $\mathcal L_{\rm ord}$, and a symmetry-consistency term $\mathcal L_{\rm inv}$. 
The resulting objective is
$
\mathcal L
=
\mathcal L_{\rm PPO}
+c_v\mathcal L_V
+\beta_g\mathcal L_g
+\beta_{\rm ord}\mathcal L_{\rm ord}
+\beta_{\rm inv}\mathcal L_{\rm inv}
-\beta_H\E[H(\pi_\theta(\cdot\mid x))].
$
Here,
$
\mathcal L_g=\E[\|\widehat G_\theta(x,a)-G^\pi(x,a)\|_2^2],
$
$
\mathcal L_{\rm ord}
=
\E\!\left[\sum_{a\succ b}
[m+\mathcal E_\theta(x,a)-\mathcal E_\theta(x,b)]_+\right],
$
and
$
\mathcal L_{\rm inv}
=
\E_{\gamma\in\Gamma_x}
[(\mathcal E_\theta(\gamma x,\gamma a)-\mathcal E_\theta(x,a))^2].
$

\begin{algorithm}[t]
\caption{\LNOQRD{} decision step}
\label{alg:qold}
\begin{algorithmic}[1]
\REQUIRE state--intent pair $x$, candidate budget $N$, frontier budget $K_{\rm fr}$, parameters $\theta$, multiplier $\lambda$
\STATE $\mathcal C \gets \textsc{GenerateCandidates}(x,N)$
\STATE $\mathcal S \gets \{a\in\mathcal C:\mathsf{Res}_{\rm law}(x,a)\le\varepsilon\}$
\IF{$\mathcal S=\emptyset$}
    \STATE $\mathcal B \gets \{\bot\}$
\ELSE
    \STATE $\mathcal Q \gets \textsc{Quotient}(\mathcal S,\Gamma_x)$ using exact or signature-based state--intent canonicalization
    \STATE $\mathcal D \gets \textsc{NonDominated}(\mathcal Q,\lambda)$ using exact or certified approximate dominance
    \STATE $\mathcal B \gets \operatorname{Top}^{\psi_\lambda}_{K}(\{\operatorname{rep}_x(c):c\in\mathcal D\})$
\ENDIF
\STATE Choose $a_t\sim\pi_\theta(\cdot\mid x)$ on $\mathcal B$ during training, or $a_t\in\arg\min_{a\in\mathcal B}\mathcal E_\theta(x,a)$ during execution
\STATE Revalidate or repair $a_t$ under the exact executor laws; map an unsuccessful repair to $\bot$
\STATE Execute the resulting action and update $\theta,\lambda$ from reward, soft costs, residuals, and next state
\ENSURE frontier $\mathcal B_\lambda(x)=\mathcal B$ and not-to-optimize set $\mathcal N^{\rm no}_\lambda(x)=\mathcal C\setminus\mathcal B$
\end{algorithmic}
\end{algorithm}

\section{Performance Guarantees}
\label{sec:guarantees}

This section accounts for what is lost after \LNOQRD{} reduces the generated candidate set. Exact residual screening, exact quotienting, and exact dominance are lossless in the senses proved above, while practical \LNOQRD{} may incur loss from finite generation, tolerance, signature quotienting, approximate dominance, top-$K$ truncation, and learned energies. We first decompose frontier loss by stage, then bound the effect of intermediate primal estimates, and finally convert frontier coverage plus energy approximation into fixed-$\lambda$ policy performance.

Throughout, $\lambda\ge0$ is fixed; multiplier learning can be handled by standard primal--dual updates around this analysis \cite{altman2021constrained}. Let $\mathcal A_{\rm bench}(x)$ be a benchmark action domain containing $\mathcal C_N(x)\cup\{\bot\}$ on which $Q_\lambda^\star(x,a)$ is defined. It may be either the generated interface with fallback or a larger executable oracle domain, in which case the first loss term is the finite-generation gap.

\subsection{Stagewise Frontier Error Accounting}

The frontier is constructed through nonempty sets
$
U_0(x)=\mathcal A_{\rm bench}(x), \quad
U_1(x)=\mathcal C_N(x)\cup\{\bot\}, \quad
U_2(x)=\mathcal S_\varepsilon(x), \quad
U_3(x)=\{\operatorname{rep}_x(c):c\in\mathcal Q_\varepsilon(x)\}, \quad
U_4(x)=\{\operatorname{rep}_x(c):c\in\mathcal D_\lambda(x)\}, \quad
U_5(x)=\mathcal B_\lambda(x).
$
Here $U_1$ is the generated interface with fallback, $U_2$ is residual-screened, $U_3$ is quotient-representative, $U_4$ is non-dominated, and $U_5$ is budgeted. If screening returns $\{\bot\}$, then $U_3=U_4=U_5=\{\bot\}$.

For $i=1,\ldots,5$, define the stage loss
\begin{equation}
    \Delta_i(x)
    =
    \max_{a\in U_{i-1}(x)}Q_\lambda^\star(x,a)
    -
    \max_{a\in U_i(x)}Q_\lambda^\star(x,a).
    \label{eq:delta_i}
\end{equation}

Write $\Delta_{\rm gen}=\Delta_1$, $\Delta_{\rm res}=\Delta_2$, $\Delta_{\rm q}=\Delta_3$, $\Delta_{\rm dom}=\Delta_4$, $\Delta_K=\Delta_5$, and $\delta_{\rm fr}(x)=\sum_{i=1}^5\Delta_i(x)$.

\begin{theorem}[Frontier error accounting]
\label{thm:frontier_accounting}
For every state--intent pair $x=(s,I)$ with nonempty $U_i(x)$,
\begin{equation}
    \max_{a\in\mathcal A_{\rm bench}(x)}
    Q_\lambda^\star(x,a)
    -
    \max_{a\in\mathcal B_\lambda(x)}
    Q_\lambda^\star(x,a)
    =
    \delta_{\rm fr}(x).
    \label{eq:frontier_preservation}
\end{equation}
Moreover, exact residual screening contributes no loss from feasible generated candidates, exact quotienting gives $\Delta_{\rm q}(x)=0$, exact dominance gives $\Delta_{\rm dom}(x)=0$, and a budget rule that retains a near-best representative gives small $\Delta_K(x)$.
\end{theorem}

\begin{proof}[Proof sketch]
The identity is the telescoping sum of the five differences in \eqref{eq:delta_i}. Exact residual screening removes only non-executable generated candidates, exact quotienting gives $\Delta_{\rm q}(x)=0$ by Theorem~\ref{thm:quotient_optimality}, and exact dominance gives $\Delta_{\rm dom}(x)=0$ by Theorem~\ref{thm:dominance_reduction}. The remaining terms correspond to finite generation, tolerance or fallback effects, and budget truncation.
\end{proof}

Accordingly, any frontier-quality loss is isolated in a named stage term, motivating the reported reduction gaps, near-oracle coverage, violation, top-$K$ loss, scalability, and overhead.

\subsection{Effect of Intermediate Primal Estimation Error}

Approximate stages may rely on critic estimates, negative-energy proxies, learned soft-cost predictions, or state--intent embeddings. The next result bounds how such intermediate primal errors affect frontier loss.

Let $\widehat Q_\theta(x,a)$ denote the estimate used by approximate quotienting, dominance, or top-$K$ selection. Assume that, on all candidates entering the reduction pipeline,
\begin{equation}
    \left|
    \widehat Q_\theta(x,a)-Q_\lambda^\star(x,a)
    \right|
    \le
    \epsilon_Q .
    \label{eq:critic_uniform_error}
\end{equation}
For a reduction stage $U_{i-1}(x)\to U_i(x)$, define the estimated slack $\eta_i(x)$ by
\begin{equation}
    \max_{a\in U_{i-1}(x)}
    \widehat Q_\theta(x,a)
    -
    \max_{a\in U_i(x)}
    \widehat Q_\theta(x,a)
    \le
    \eta_i(x).
    \label{eq:estimated_stage_slack}
\end{equation}
A zero value of $\eta_i(x)$ means that the stage keeps an estimated best candidate. A positive value allows approximate screening, approximate signatures, approximate dominance, or top-$K$ truncation.

\begin{theorem}[Intermediate-signal error bound]
\label{thm:intermediate_signal_error}
Assume \eqref{eq:critic_uniform_error}. For any stage $U_{i-1}(x)\to U_i(x)$ satisfying \eqref{eq:estimated_stage_slack},
\begin{equation}
    \Delta_i(x)
    \le
    \eta_i(x)+2\epsilon_Q .
    \label{eq:stage_error_bound}
\end{equation}
Consequently, if the approximate reduction stages $i=2,3,4,5$ satisfy \eqref{eq:estimated_stage_slack}, then
\begin{equation}
    \delta_{\rm fr}(x)
    \le
    \Delta_{\rm gen}(x)
    +
    \sum_{i=2}^{5}\eta_i(x)
    +
    8\epsilon_Q .
    \label{eq:frontier_error_from_primal_signal}
\end{equation}
Exact lossless stages contribute zero to the corresponding true loss $\Delta_i(x)$; their estimated slack $\eta_i(x)$ need not vanish.
\end{theorem}

\begin{proof}[Proof sketch]
Fix a stage $U_{i-1}\to U_i$. Let
$
    a^\star
    \in
    \operatorname*{argmax}_{a\in U_{i-1}(x)}
    Q_\lambda^\star(x,a),
    \qquad
    \bar a
    \in
    \operatorname*{argmax}_{a\in U_i(x)}
    \widehat Q_\theta(x,a).
$
By \eqref{eq:critic_uniform_error},
$
    Q_\lambda^\star(x,a^\star)
    \le
    \widehat Q_\theta(x,a^\star)+\epsilon_Q
    \le
    \max_{a\in U_{i-1}(x)}\widehat Q_\theta(x,a)+\epsilon_Q.
$
By \eqref{eq:estimated_stage_slack},
$
    \max_{a\in U_{i-1}(x)}\widehat Q_\theta(x,a)
    \le
    \widehat Q_\theta(x,\bar a)+\eta_i(x).
$
Again using \eqref{eq:critic_uniform_error},
$
    \widehat Q_\theta(x,\bar a)
    \le
    Q_\lambda^\star(x,\bar a)+\epsilon_Q
    \le
    \max_{a\in U_i(x)}Q_\lambda^\star(x,a)+\epsilon_Q.
$
Combining the three inequalities gives \eqref{eq:stage_error_bound}. Summing the four approximate-stage bounds and adding the generation term gives \eqref{eq:frontier_error_from_primal_signal}.
\end{proof}

This theorem identifies the explicit error channel created by learning not to optimize. If the primal estimates are accurate and each approximate reduction stage keeps a candidate with near-best estimated value, the true frontier loss remains small. If a signature quotient or approximate dominance rule removes a high-value candidate, the loss appears in the corresponding $\eta_i(x)$ term.

\subsection{Reduced-Policy Performance Bound}

We now convert frontier coverage and energy approximation into a value bound. Let $N_{\max}=\sup_x|\mathcal B_\lambda(x)|<\infty$; by Theorem~\ref{thm:frontier_size}, $N_{\max}\le K_{\rm fr}$ when an executable candidate is found and $N_{\max}=1$ under fallback.

\begin{assumption}[Uniform reduced-frontier coverage]
\label{ass:approx_coverage}
There exists $\delta\ge0$ such that, for every state $x$,
$
\max_{a\in\mathcal A_{\rm bench}(x)}Q_\lambda^\star(x,a)
-
\max_{a\in\mathcal B_\lambda(x)}Q_\lambda^\star(x,a)
\le\delta .
$
\end{assumption}
Theorem~\ref{thm:frontier_accounting} decomposes $\delta$ by reduction stage, and Theorem~\ref{thm:intermediate_signal_error} bounds the contribution of approximate primal signals.

\begin{assumption}[Uniform energy approximation]
\label{ass:energy_approx}
There exists $\epsilon\ge0$ such that, for every $x$ and $a\in\mathcal B_\lambda(x)$,
$
|-\mathcal E_\theta(x,a)-Q_\lambda^\star(x,a)|\le\epsilon .
$
\end{assumption}

\begin{theorem}[Approximate reduced-policy performance]
\label{thm:approx_reduced_policy}
Under Assumptions~\ref{ass:approx_coverage} and~\ref{ass:energy_approx}, the Boltzmann policy \eqref{eq:energy_policy_revised} satisfies, for every initial state $x$,
\begin{equation}
    V_\lambda^\star(x)-V_\lambda^{\pi_\theta}(x)
    \le
    \frac{
    \delta+2\epsilon+\tau\log N_{\max}
    }{1-\gamma}.
    \label{eq:performance_bound_revised}
\end{equation}
\end{theorem}

\begin{proof}[Proof sketch]
Fix $x$ and write $\widehat Q_\theta(x,a)=-\mathcal E_\theta(x,a)$. Assumption~\ref{ass:approx_coverage} bounds the loss from restricting the benchmark domain to $\mathcal B_\lambda(x)$ by $\delta$. Assumption~\ref{ass:energy_approx} gives $|\widehat Q_\theta-Q_\lambda^\star|\le\epsilon$ on the reduced frontier. Since \eqref{eq:energy_policy_revised} is the softmax policy over $\widehat Q_\theta$, the log-sum-exp bound over at most $N_{\max}$ actions gives
$
\max_{a\in\mathcal B_\lambda(x)}\widehat Q_\theta(x,a)
-
\E_{a\sim\pi_\theta(\cdot\mid x)}
[\widehat Q_\theta(x,a)]
\le
\tau\log N_{\max}.
$
Using the uniform approximation once for the maximum and once for the expectation yields a one-step Bellman suboptimality of at most $\delta+2\epsilon+\tau\log N_{\max}$. The discounted performance-difference inequality then gives \eqref{eq:performance_bound_revised}.
\end{proof}

\begin{corollary}[Exact-structure limit]
\label{cor:exact_structure_limit}
If the generator contains an optimal benchmark action, residual screening does not remove feasible optimal actions, quotienting and dominance are exact, top-$K$ preserves an optimal non-dominated representative, and the energy represents $-Q_\lambda^\star$ exactly on the reduced frontier, then $\delta=\epsilon=0$. A deterministic zero-temperature policy over $\mathcal B_\lambda(x)$ is optimal for the fixed-$\lambda$ MDP.
\end{corollary}

\begin{proof}[Proof sketch]
The stated exact conditions make the frontier loss zero by Theorem~\ref{thm:frontier_accounting} and set $\epsilon=0$ in Assumption~\ref{ass:energy_approx}. Taking the deterministic zero-temperature limit removes the entropy term in Theorem~\ref{thm:approx_reduced_policy}.
\end{proof}

\section{Experiments}
\label{sec:experiments}

This section evaluates whether \LNOQRD{} improves intent-conditioned network deployment by learning what not to optimize. The controlled benchmarks test two claims: Theorem~\ref{thm:frontier_size} predicts that quotienting, residual screening, dominance, and the top-$K$ budget reduce the candidates exposed to primal ranking when saved scoring cost exceeds dual overhead; Theorem~\ref{thm:frontier_accounting} and Theorem~\ref{thm:intermediate_signal_error} predict that quality loss can be traced to generation, screening, quotienting, dominance, budget truncation, or intermediate primal-estimation error. We therefore report utility, reduction, near-oracle coverage, reduction--gap behavior, violation, scalability, transfer, ablation, and post-generation latency.

\begin{table}[t]
\centering
\caption{Benchmark settings.}
\label{tab:benchmarks}
\scriptsize
\setlength{\tabcolsep}{1.5pt}
\renewcommand{\arraystretch}{0.88}
\resizebox{0.88\linewidth}{!}{%
\begin{tabular}{@{}p{0.23\linewidth}p{0.14\linewidth}p{0.22\linewidth}p{0.31\linewidth}@{}}
\toprule
\textbf{Scenario} & \textbf{Scale} & \textbf{Intent} & \textbf{Main stressor}\\
\midrule
Service graph & Small/Large & DAG service graph & Placement--routing coupling\\
Network slice & Small/Large & Latency, isolation, reliability & Multi-tenant contention\\
Service pipeline & Large & Ordered pipeline DAG & Unseen service-graph transfer\\
Dynamic drift & Large & Time-varying intents & Queue, demand, link-rate, and reliability variation\\
\bottomrule
\end{tabular}%
}
\vspace{-1.5mm}
\end{table}

\begin{table}[t]
\centering
\caption{Small-instance explanatory results. Values are mean $\pm$ standard deviation.}
\label{tab:explanatory_results}
\scriptsize
\setlength{\tabcolsep}{1.8pt}
\renewcommand{\arraystretch}{0.88}
\resizebox{0.88\linewidth}{!}{%
\begin{tabular}{@{}lcccc@{}}
\toprule
\textbf{Method} & \textbf{Reduction} & \textbf{Coverage} & \textbf{OracleGap} & \textbf{Violation}\\
\midrule
Raw candidates & $0.011\pm0.015$ & $0.986\pm0.013$ & $0.021\pm0.007$ & $0.268\pm0.010$\\
Quotient only & $0.338\pm0.025$ & $0.957\pm0.015$ & $0.036\pm0.008$ & $0.196\pm0.010$\\
Quotient + dominance & $0.579\pm0.026$ & $0.934\pm0.016$ & $0.056\pm0.007$ & $0.123\pm0.010$\\
Full \LNOQRD{} & $\mathbf{0.759\pm0.025}$ & $0.908\pm0.016$ & $0.076\pm0.008$ & $\mathbf{0.064\pm0.010}$\\
\bottomrule
\end{tabular}%
}
\vspace{-1.5mm}
\end{table}

\begin{figure}[t]
\centering
\begin{minipage}[t]{0.49\linewidth}
\centering
\includegraphics[width=\linewidth]{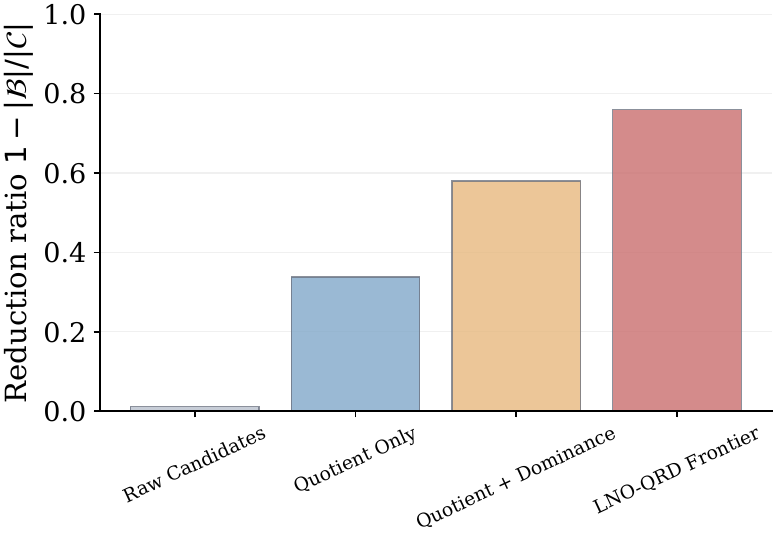}
\vspace{-1mm}
{\footnotesize (a) Candidate reduction.}
\end{minipage}\hfill
\begin{minipage}[t]{0.49\linewidth}
\centering
\includegraphics[width=\linewidth]{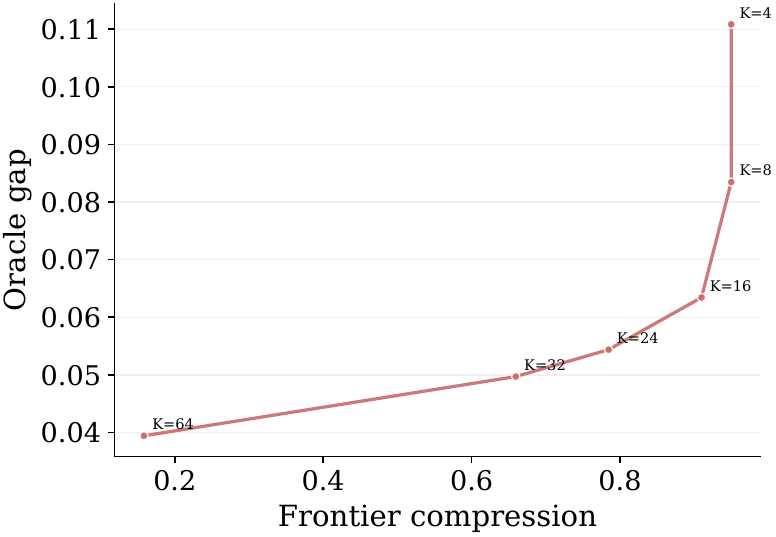}
\vspace{-1mm}
{\footnotesize (b) Reduction--gap tradeoff.}
\end{minipage}
\caption{Small-instance frontier diagnostics. \LNOQRD{} progressively shrinks the candidate set while keeping the oracle gap controlled.}
\label{fig:frontier_diagnostics}
\vspace{-1mm}
\end{figure}

\begin{figure}[t]
\centering
\includegraphics[width=0.55\linewidth]{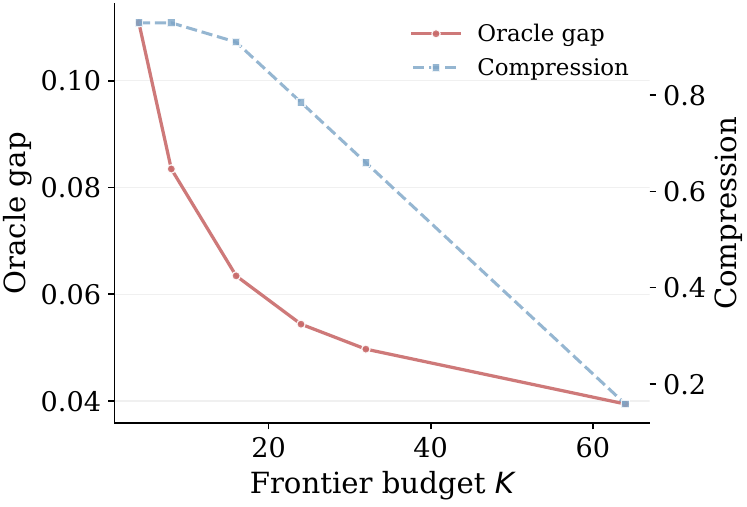}
\caption{Budget-induced reduction--gap tradeoff. A larger frontier budget lowers $\Delta_K$ but exposes more candidates to primal scoring.}
\label{fig:budget_reduction_gap}
\vspace{-1mm}
\end{figure}

\begin{table*}[t]
\centering
\caption{Large-scale deployment results. Values are mean $\pm$ standard deviation over four large benchmark families, three seeds, and all test episodes.}
\label{tab:large_results}
\scriptsize
\setlength{\tabcolsep}{2.2pt}
\renewcommand{\arraystretch}{0.70}
\resizebox{0.70\textwidth}{!}{%
\begin{tabular}{@{}lccccc@{}}
\toprule
\textbf{Method} & \textbf{Utility} & \textbf{IntentSat} & \textbf{Violation} & \textbf{Reduction} & \textbf{Time (ms)}\\
\midrule
Greedy & $3.001\pm0.073$ & $0.277\pm0.059$ & $0.500\pm0.024$ & $0.010\pm0.014$ & $10.537\pm0.401$\\
MLP-PPO & $3.305\pm0.073$ & $0.374\pm0.062$ & $0.431\pm0.022$ & $0.010\pm0.015$ & $22.896\pm0.463$\\
GNN-PPO & $3.548\pm0.074$ & $0.486\pm0.059$ & $0.355\pm0.021$ & $0.011\pm0.015$ & $25.186\pm0.488$\\
Transformer & $3.642\pm0.070$ & $0.525\pm0.058$ & $0.327\pm0.020$ & $0.009\pm0.014$ & $29.327\pm0.488$\\
\textsc{Penalty-RL} & $3.523\pm0.072$ & $0.579\pm0.056$ & $0.291\pm0.019$ & $0.009\pm0.014$ & $24.029\pm0.448$\\
\LNOQRD{} & $\mathbf{3.973\pm0.071}$ & $\mathbf{0.724\pm0.056}$ & $\mathbf{0.191\pm0.017}$ & $\mathbf{0.730\pm0.025}$ & $\mathbf{7.008\pm0.459}$\\
\bottomrule
\end{tabular}%
}
\vspace{-1.5mm}
\end{table*}

\begin{figure}[t]
\centering
\begin{minipage}[t]{0.49\linewidth}
\centering
\includegraphics[width=\linewidth]{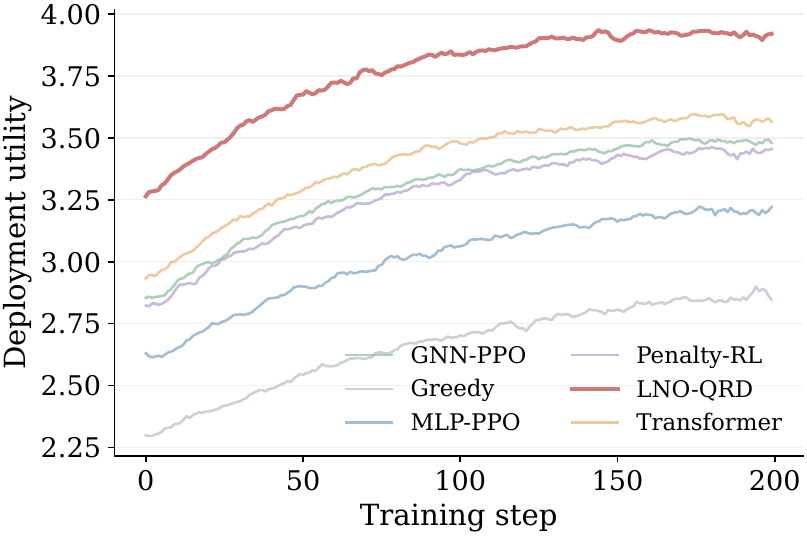}
\vspace{-1mm}
{\footnotesize (a) Learning curves.}
\end{minipage}\hfill
\begin{minipage}[t]{0.49\linewidth}
\centering
\includegraphics[width=\linewidth]{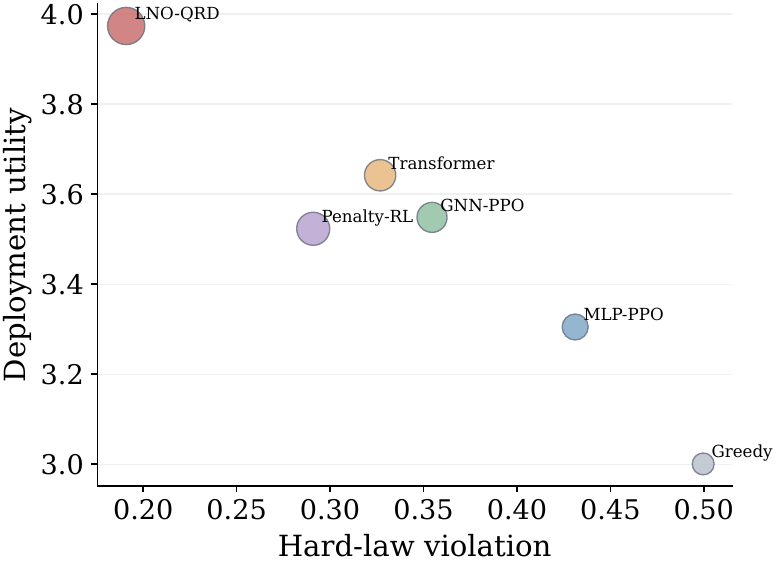}
\vspace{-1mm}
{\footnotesize (b) Utility--violation tradeoff.}
\end{minipage}
\caption{Large-scale deployment performance. \LNOQRD{} improves utility while shifting the tradeoff toward lower violation.}
\label{fig:large_performance}
\vspace{-1mm}
\end{figure}

\begin{figure}[t]
\centering
\begin{minipage}[t]{0.49\linewidth}
\centering
\includegraphics[width=\linewidth]{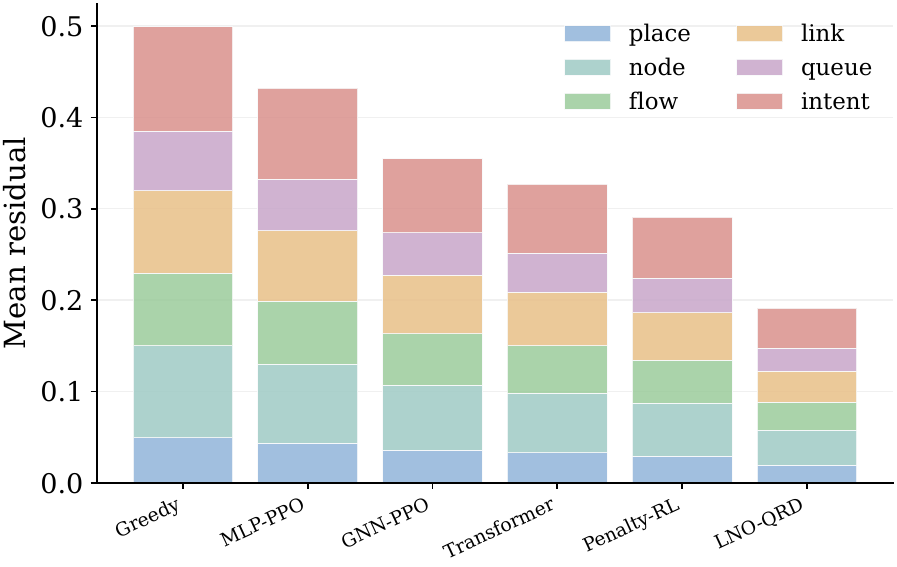}
\vspace{-1mm}
{\footnotesize (a) Law-residual profile.}
\end{minipage}\hfill
\begin{minipage}[t]{0.49\linewidth}
\centering
\includegraphics[width=\linewidth]{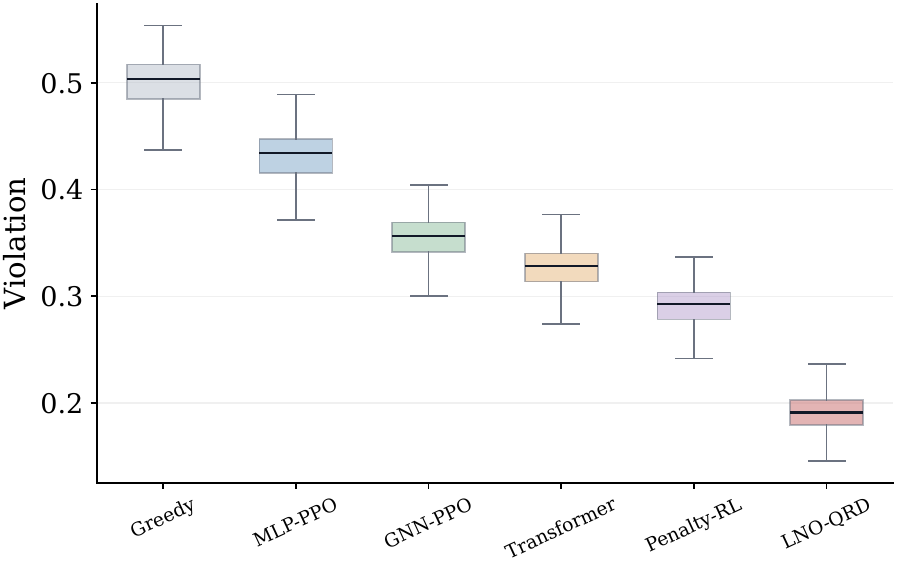}
\vspace{-1mm}
{\footnotesize (b) Violation spread.}
\end{minipage}
\caption{Feasibility diagnostics. \LNOQRD{} reduces multiple residual coordinates and lowers violation spread across large scenarios.}
\label{fig:law_feasibility}
\vspace{-1mm}
\end{figure}

\begin{figure}[t]
\centering
\begin{minipage}[t]{0.49\linewidth}
\centering
\includegraphics[width=\linewidth]{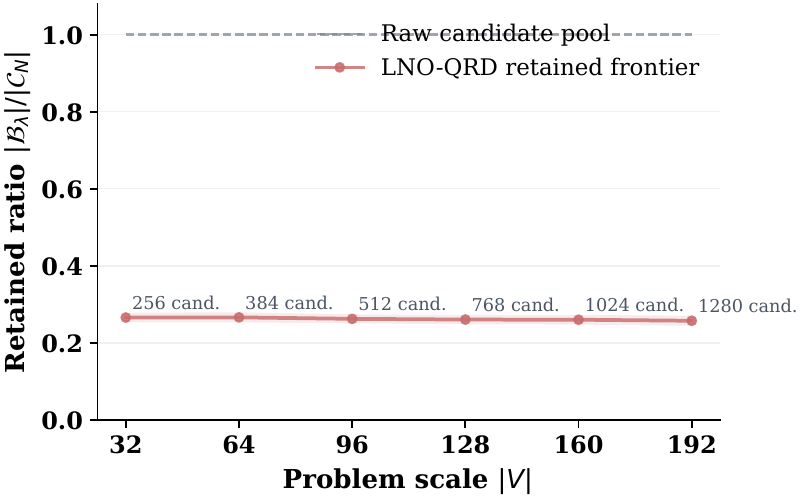}
\vspace{-1mm}
{\footnotesize (a) Frontier ratio vs. scale.}
\end{minipage}\hfill
\begin{minipage}[t]{0.49\linewidth}
\centering
\includegraphics[width=\linewidth]{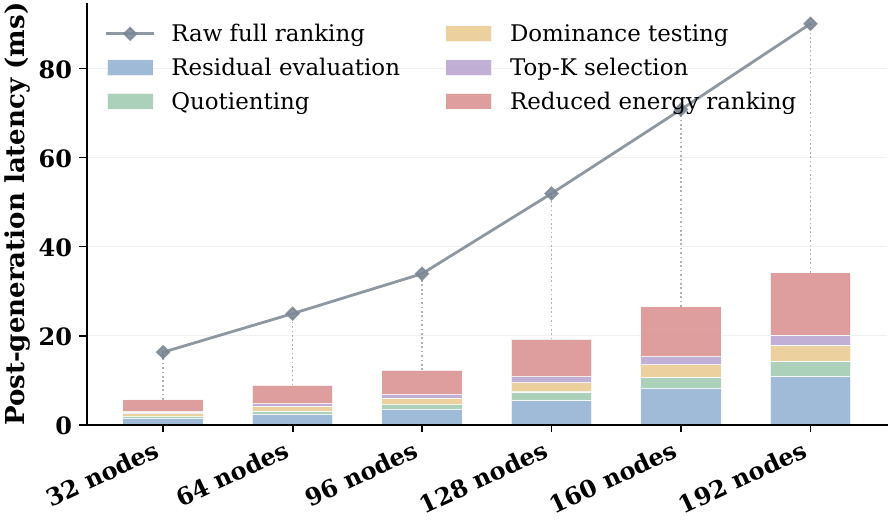}
\vspace{-1mm}
{\footnotesize (b) Overhead decomposition.}
\end{minipage}
\caption{Scalability and overhead. Panel (a) shows the retained-frontier ratio as scale grows. Panel (b) decomposes post-generation latency and compares the reduced pipeline with raw full-candidate ranking.}
\label{fig:scalability_overhead}
\vspace{-1mm}
\end{figure}

\begin{table}[t]
\centering
\caption{Ablation study on the two large structural benchmarks. Values are mean $\pm$ standard deviation.}
\label{tab:ablation}
\scriptsize
\setlength{\tabcolsep}{1.8pt}
\renewcommand{\arraystretch}{0.88}
\resizebox{0.88\linewidth}{!}{%
\begin{tabular}{@{}lcccc@{}}
\toprule
\textbf{Variant} & \textbf{Utility} & \textbf{Violation} & \textbf{Reduction} & \textbf{GenGap}\\
\midrule
Full & $\mathbf{4.003\pm0.080}$ & $\mathbf{0.187\pm0.016}$ & $\mathbf{0.731\pm0.026}$ & $\mathbf{0.168\pm0.016}$\\
NoQuotient & $3.886\pm0.077$ & $0.206\pm0.018$ & $0.569\pm0.025$ & $0.175\pm0.016$\\
NoDominance & $3.935\pm0.082$ & $0.195\pm0.017$ & $0.402\pm0.025$ & $0.172\pm0.015$\\
NoResidual & $3.805\pm0.079$ & $0.276\pm0.019$ & $0.599\pm0.025$ & $0.183\pm0.015$\\
NoOrderLoss & $3.904\pm0.078$ & $0.213\pm0.018$ & $0.685\pm0.026$ & $0.176\pm0.016$\\
NoInvLoss & $3.914\pm0.080$ & $0.208\pm0.018$ & $0.701\pm0.025$ & $0.174\pm0.016$\\
\bottomrule
\end{tabular}%
}
\vspace{-1.5mm}
\end{table}

\begin{figure}[h]
\centering
\includegraphics[width=0.82\linewidth]{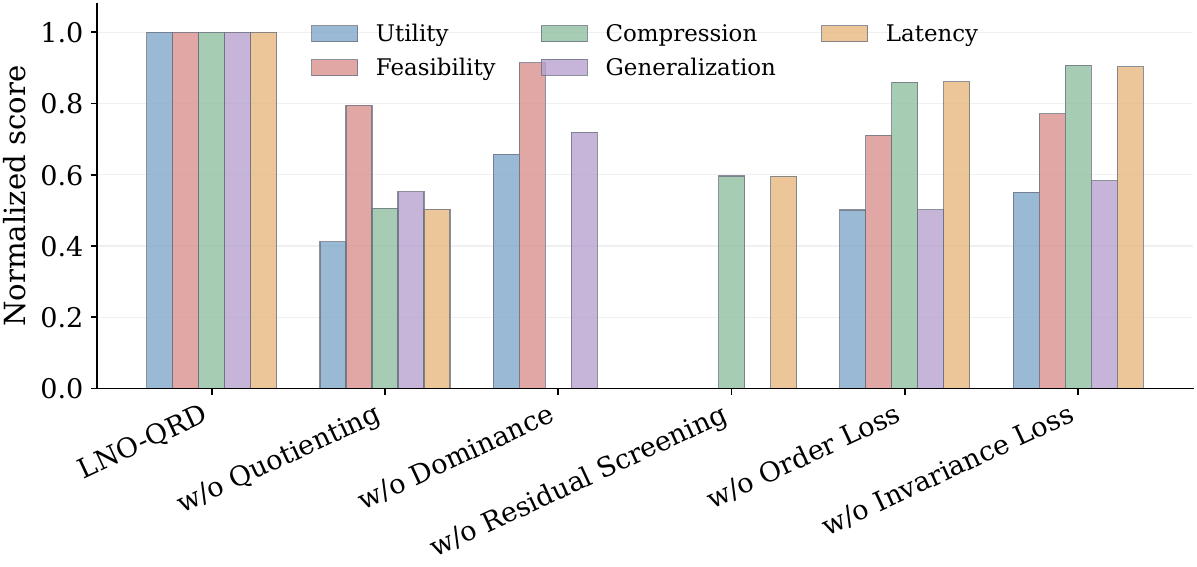}
\caption{Ablation summary. Residual screening mainly improves feasibility, quotienting removes state--intent redundancy, dominance gives the largest compression, and order/invariance losses stabilize reduced-frontier ranking.}
\label{fig:ablation_summary}
\vspace{-1mm}
\end{figure}

\subsection{Experimental Setup}
\label{subsec:exp_setup}

We evaluate four benchmark families: service-graph deployment, network-slice realization, service-pipeline deployment, and dynamic drift \cite{mijumbi2015network,bhamare2016survey,sun2022survey}. Small instances use $|V|\in\{8,10\}$, $|K|\in\{4,5\}$, $N=64$, and frontier budget $K_{\rm fr}=16$ over $10$ seeds for oracle diagnostics. Large instances use $|V|\in[124,128]$, $|K|\in[10,14]$, $N=512$, and frontier budget $K_{\rm fr}=24$; dynamic drift varies traffic, queues, link rates, and reliability. All candidate-based methods share the same $\mathcal C_N(x)$ generator, including capacity-, load-, locality-, trust-, delay-, bandwidth-, reliability-aware, balanced, randomized, and repair-based candidates. Hence any gain must come from excluding infeasible, equivalent, dominated, or unnecessary candidates, not from a stronger generator.

\subsection{Baselines and Metrics}
\label{subsec:exp_baselines_metrics}

We compare \LNOQRD{} with six baselines: a small-instance MILP/CP oracle \cite{boyd2004convex,papadimitriou1998combinatorial}, Greedy, MLP-PPO \cite{schulman2017proximal}, GNN-PPO \cite{kipf2016semi,velivckovic2017graph}, Transformer candidate ranking \cite{vaswani2017attention,lee2019set}, and \textsc{Penalty-RL} \cite{altman2021constrained,achiam2017constrained}. The ablations remove quotienting, dominance pruning, law residuals, order supervision, or invariance supervision.

Utility is the deployment reward after evaluating delay, reliability, queue, and resource terms; IntentSat is normalized to $[0,1]$; Violation is the normalized pre-revalidation residual of the selected proposal; Reduction is $1-|\mathcal B_\lambda(x)|/|\mathcal C_N(x)|$. OracleGap is measured against the small-instance oracle or a large-pool proxy, Coverage is the preserved fraction of near-oracle deployments, GenGap is the unseen-intent/topology transfer gap, and Time is post-generation latency for candidate evaluation, reduction, and ranking.

\subsection{Small-Instance Explanatory Results}
\label{subsec:exp_explanatory}

Small certifiable instances test the error-accounting view in Theorem~\ref{thm:frontier_accounting}. Raw candidates preserve high oracle coverage but contain many law-violating deployments. Quotienting reduces the set by $33.8\%$ with $95.7\%$ coverage; adding dominance raises reduction to $57.9\%$; full \LNOQRD{} reaches $75.9\%$ reduction and the lowest violation while retaining $90.8\%$ near-oracle coverage. Thus most candidates are removed before primal ranking, while the high-value frontier remains largely covered.

\subsection{Large-Scale Deployment Results}
\label{subsec:exp_large}

Large-scale results test whether frontier reduction improves policy quality and online latency. Across four large benchmark families, \LNOQRD{} obtains the best utility, IntentSat, violation, reduction, and post-generation time in Table~\ref{tab:large_results}. Relative to Transformer, utility increases from $3.642$ to $3.973$ and latency drops from $29.327$ ms to $7.008$ ms; relative to \textsc{Penalty-RL}, violation decreases from $0.291$ to $0.191$. With average reduction $0.730$, the primal learner ranks only about $27\%$ of the raw candidate pool.

\subsection{Feasibility and Law-Residual Analysis}
\label{subsec:exp_law_residual}

Residual diagnostics explain the feasibility gain. \textsc{Penalty-RL} reduces violations through training penalties, whereas \LNOQRD{} uses law residuals as pre-ranking not-to-optimize certificates. Mean violation drops from $0.291$ for \textsc{Penalty-RL} to $0.191$ for \LNOQRD{}, with lower placement, node-capacity, flow, link-capacity, queue, and intent residuals in Fig.~\ref{fig:law_feasibility}.

\subsection{Scalability and Dual-Learning Overhead}
\label{subsec:exp_scalability_overhead}

Scalability and overhead results test the speedup condition in Theorem~\ref{thm:frontier_size} and \eqref{eq:lno_speedup_condition}. As scale grows, the retained-frontier ratio $|\mathcal B_\lambda(x)|/|\mathcal C_N(x)|$ remains substantially below one, so the primal scorer is applied to a compact subset rather than the full candidate pool. The overhead decomposition further shows that residual evaluation, quotienting, dominance testing, and top-$K$ selection are outweighed by the saved full-candidate energy-ranking cost.

\subsection{Transfer Across Intents and Topologies}
\label{subsec:exp_transfer}

For transfer, we train on service-graph or network-slice deployments and test on unseen large benchmark families. We report aggregate behavior to save space: \LNOQRD{} achieves the highest average transfer utility, $3.908$, compared with $3.577$ for Transformer and $3.490$ for GNN-PPO. Same-family transfer is easier, while dynamic drift and unseen service pipelines produce larger gaps; nevertheless, residual, quotient, and dominance certificates transfer better than fixed candidate-index policies because they depend on network laws and state--intent structure.

\subsection{Ablation Study}
\label{subsec:exp_ablation}

The ablation study isolates the certificates. Removing quotienting weakens utility and compression, confirming the role of state--intent equivalence. Removing dominance causes the largest reduction loss, showing that it excludes many feasible but noncompetitive candidates. Removing residual screening gives the largest feasibility degradation, while removing order or invariance losses weakens energy ranking on the reduced frontier.

Overall, the experiments support the central claim: \LNOQRD{} reduces small-instance candidates by $75.9\%$ while retaining $90.8\%$ near-oracle coverage, and achieves the best large-scale utility, violation, reduction, and ranking latency. Scalability, transfer, and ablation results attribute these gains to residual screening, quotienting, dominance, and structural energy learning rather than to the shared generator.

\section{Conclusion}
This paper studied learning not to optimize as a complementary counterpart to constrained network policy optimization. \LNOQRD{} uses residual screening, state--intent quotienting, and dominance pruning to remove infeasible, redundant, and dominated candidates before expensive primal ranking, so the learner optimizes only over a budgeted frontier. We proved lossless reduction conditions, a frontier-size and ranking-speedup bound, and stagewise error accounts for approximate certificates and intermediate learning signals. Experiments on controlled deployment benchmarks show that this dual layer preserves near-oracle candidates on small instances and improves large-scale utility, violation, transfer, and ranking latency. The main conclusion is that intermediate network-law information can shape not only how actions are ranked, but also which actions should not be optimized.

\nocite{*}
\bibliographystyle{IEEEtran}
\bibliography{refs}

@article{clemm2022intent,
  title={Intent-based networking-concepts and definitions},
  author={Clemm, Alexander and Ciavaglia, Laurent and Granville, Lisandro Zambenedetti and Tantsura, Jeff},
  year={2022},
  publisher={IETF RFC 9315, Oct}
}

@article{li2022intent,
  title={Intent classification},
  author={Li, Chen and Havel, Olga and Olariu, Aiana and Martinez-Julia, Peo and Nobre, J{\'e}ferson Campos and Lopez, Diego},
  journal={RFC 9316},
  year={2022}
}

@article{yao2025use,
  title={Use Cases and Practices for Intent-Based Networking},
  author={Yao, K and Chen, D and Jeong, JP and Wu, Q and Yang, C and Contreras, LM and Fioccola, G},
  journal={Internet Engineering Task Force, Internet-Draft draft-irtf-nmrg-ibnusecases-00},
  year={2025}
}

@article{leivadeas2022survey,
  title={A survey on intent-based networking},
  author={Leivadeas, Aris and Falkner, Matthias},
  journal={IEEE Communications Surveys \& Tutorials},
  volume={25},
  number={1},
  pages={625--655},
  year={2022},
  publisher={IEEE}
}

@article{zhang2024distributed,
  title={Distributed age-of-information scheduling with noma via deep reinforcement learning},
  author={Zhang, Congwei and Zou, Yifei and Zhang, Zuyuan and Yu, Dongxiao and G{\'o}mez, Jorge Torres and Lan, Tian and Dressler, Falko and Cheng, Xiuzhen},
  journal={IEEE Transactions on Mobile Computing},
  volume={24},
  number={1},
  pages={30--44},
  year={2024},
  publisher={IEEE}
}

@article{zou2024distributed,
  title={A distributed abstract mac layer for cooperative learning on internet of vehicles},
  author={Zou, Yifei and Zhang, Zuyuan and Zhang, Congwei and Zheng, Yanwei and Yu, Dongxiao and Yu, Jiguo},
  journal={IEEE Transactions on Intelligent Transportation Systems},
  volume={25},
  number={8},
  pages={8972--8983},
  year={2024},
  publisher={IEEE}
}

@inproceedings{zhang2025network,
  title={Network diffuser for placing-scheduling service function chains with inverse demonstration},
  author={Zhang, Zuyuan and Aggarwal, Vaneet and Lan, Tian},
  booktitle={IEEE INFOCOM 2025-IEEE Conference on Computer Communications},
  pages={1--10},
  year={2025},
  organization={IEEE}
}

@article{zhang2026operator,
  title={Operator-guided invariance learning for continuous reinforcement learning},
  author={Zhang, Zuyuan and Yu, Fei Xu and Lan, Tian},
  journal={arXiv preprint arXiv:2605.06500},
  year={2026}
}

@inproceedings{zhanggeometric,
  title={Geometric Coherence Learning for Structuring Value Functions in Plain MDPs},
  author={Zhang, Zuyuan and Fang, Zeyu and Lan, Tian},
  booktitle={Forty-third International Conference on Machine Learning}
}

@article{zhang2026geometry,
  title={Geometry of drifting mdps with path-integral stability certificates},
  author={Zhang, Zuyuan and Imani, Mahdi and Lan, Tian},
  journal={arXiv preprint arXiv:2601.21991},
  year={2026}
}

@inproceedings{qiao2024br,
  title={Br-defedrl: Byzantine-robust decentralized federated reinforcement learning with fast convergence and communication efficiency},
  author={Qiao, Jing and Zhang, Zuyuan and Yue, Sheng and Yuan, Yuan and Cai, Zhipeng and Zhang, Xiao and Ren, Ju and Yu, Dongxiao},
  booktitle={Ieee infocom 2024-ieee conference on computer communications},
  pages={141--150},
  year={2024},
  organization={IEEE}
}

@inproceedings{yu2024look,
  title={Look-ahead robust network optimization with generative state predictions},
  author={Yu, Fei Xu and Zhang, Zuyuan and Grob, Emily and Adam, Gina and Coffey, Sean and Bastian, Nathaniel D and Lan, Tian},
  booktitle={AAAI 2025 Workshop on Artificial Intelligence for Wireless Communications and Networking (AI4WCN)},
  year={2024}
}

@inproceedings{zhanghodgeflow,
  title={HodgeFlow Policy Search by Topologically Dissecting Temporal-Difference Signals in Non-Markovian Environments},
  author={Zhang, Zuyuan and Tang, Sizhe and Lan, Tian},
  booktitle={Forty-third International Conference on Machine Learning}
}

@inproceedings{zhang2026lisfc,
  title={Lisfc-search: Lifelong search for network sfc optimization under non-stationary drifts},
  author={Zhang, Zuyuan and Aggarwal, Vaneet and Lan, Tian},
  booktitle={IEEE INFOCOM 2026-IEEE Conference on Computer Communications},
  pages={1--6},
  year={2026},
  organization={IEEE}
}

@article{tang2026nonzero,
  title={Nonzero: Interaction-guided exploration for multi-agent monte carlo tree search},
  author={Tang, Sizhe and Zhang, Zuyuan and Imani, Mahdi and Lan, Tian},
  journal={arXiv preprint arXiv:2605.00751},
  year={2026}
}

@article{zhang2026counterfactual,
  title={Counterfactual Regret Minimization-Mixing for Noncooperative Stochastic Spectrum Games with Imperfect Information},
  author={Zhang, Zuyuan and Liu, Lingjia and Bastian, Nathaniel D and Lan, Tian},
  journal={IEEE Transactions on Networking},
  year={2026},
  publisher={IEEE}
}

@article{zhang2026metric,
  title={Metric-Gradient Projection for Stable Multi-Agent Policy Learning},
  author={Zhang, Zuyuan and Tang, Sizhe and Imani, Mahdi and Lan, Tian},
  journal={arXiv preprint arXiv:2605.18809},
  year={2026}
}

@article{mehmood2023intent,
  title={Intent-driven autonomous network and service management in future cellular networks: A structured literature review},
  author={Mehmood, Kashif and Kralevska, Katina and Palma, David},
  journal={Computer Networks},
  volume={220},
  pages={109477},
  year={2023},
  publisher={Elsevier}
}

@article{mijumbi2015network,
  title={Network function virtualization: State-of-the-art and research challenges},
  author={Mijumbi, Rashid and Serrat, Joan and Gorricho, Juan-Luis and Bouten, Niels and De Turck, Filip and Boutaba, Raouf},
  journal={IEEE Communications surveys \& tutorials},
  volume={18},
  number={1},
  pages={236--262},
  year={2015},
  publisher={IEEE}
}

@article{bhamare2016survey,
  title={A survey on service function chaining},
  author={Bhamare, Deval and Jain, Raj and Samaka, Mohammed and Erbad, Aiman},
  journal={Journal of Network and Computer Applications},
  volume={75},
  pages={138--155},
  year={2016},
  publisher={Elsevier}
}

@article{sun2022survey,
  title={A survey on the placement of virtual network functions},
  author={Sun, Jie and Zhang, Yi and Liu, Feng and Wang, Huandong and Xu, Xiaojian and Li, Yong},
  journal={Journal of Network and Computer Applications},
  volume={202},
  pages={103361},
  year={2022},
  publisher={Elsevier}
}

@book{papadimitriou1998combinatorial,
  title={Combinatorial optimization: algorithms and complexity},
  author={Papadimitriou, Christos H and Steiglitz, Kenneth},
  year={1998},
  publisher={Courier Corporation}
}

@inproceedings{mao2016resource,
  title={Resource management with deep reinforcement learning},
  author={Mao, Hongzi and Alizadeh, Mohammad and Menache, Ishai and Kandula, Srikanth},
  booktitle={Proceedings of the 15th ACM workshop on hot topics in networks},
  pages={50--56},
  year={2016}
}

@incollection{mao2019learning,
  title={Learning scheduling algorithms for data processing clusters},
  author={Mao, Hongzi and Schwarzkopf, Malte and Venkatakrishnan, Shaileshh Bojja and Meng, Zili and Alizadeh, Mohammad},
  booktitle={Proceedings of the ACM special interest group on data communication},
  pages={270--288},
  year={2019}
}

@article{luong2019applications,
  title={Applications of deep reinforcement learning in communications and networking: A survey},
  author={Luong, Nguyen Cong and Hoang, Dinh Thai and Gong, Shimin and Niyato, Dusit and Wang, Ping and Liang, Ying-Chang and Kim, Dong In},
  journal={IEEE communications surveys \& tutorials},
  volume={21},
  number={4},
  pages={3133--3174},
  year={2019},
  publisher={IEEE}
}

@article{xiao2021leveraging,
  title={Leveraging deep reinforcement learning for traffic engineering: A survey},
  author={Xiao, Yang and Liu, Jun and Wu, Jiawei and Ansari, Nirwan},
  journal={IEEE Communications Surveys \& Tutorials},
  volume={23},
  number={4},
  pages={2064--2097},
  year={2021},
  publisher={IEEE}
}

@book{ravindran2004algebraic,
  title={An algebraic approach to abstraction in reinforcement learning},
  author={Ravindran, Balaraman},
  year={2004},
  publisher={University of Massachusetts Amherst}
}

@article{givan2003equivalence,
  title={Equivalence notions and model minimization in Markov decision processes},
  author={Givan, Robert and Dean, Thomas and Greig, Matthew},
  journal={Artificial intelligence},
  volume={147},
  number={1-2},
  pages={163--223},
  year={2003},
  publisher={Elsevier}
}

@article{taylor2008bounding,
  title={Bounding performance loss in approximate MDP homomorphisms},
  author={Taylor, Jonathan and Precup, Doina and Panagaden, Prakash},
  journal={Advances in Neural Information Processing Systems},
  volume={21},
  year={2008}
}

@article{van2020mdp,
  title={Mdp homomorphic networks: Group symmetries in reinforcement learning},
  author={Van der Pol, Elise and Worrall, Daniel and van Hoof, Herke and Oliehoek, Frans and Welling, Max},
  journal={Advances in Neural Information Processing Systems},
  volume={33},
  pages={4199--4210},
  year={2020}
}

@book{puterman2014markov,
  title={Markov decision processes: discrete stochastic dynamic programming},
  author={Puterman, Martin L},
  year={2014},
  publisher={John Wiley \& Sons}
}

@article{smith2002structural,
  title={Structural properties of stochastic dynamic programs},
  author={Smith, James E and McCardle, Kevin F},
  journal={Operations Research},
  volume={50},
  number={5},
  pages={796--809},
  year={2002},
  publisher={INFORMS}
}

@book{topkis1998supermodularity,
  title={Supermodularity and complementarity},
  author={Topkis, Donald M},
  year={1998},
  publisher={Princeton university press}
}

@article{jiang2015approximate,
  title={An approximate dynamic programming algorithm for monotone value functions},
  author={Jiang, Daniel R and Powell, Warren B},
  journal={Operations research},
  volume={63},
  number={6},
  pages={1489--1511},
  year={2015},
  publisher={INFORMS}
}

@article{zaheer2017deep,
  title={Deep sets},
  author={Zaheer, Manzil and Kottur, Satwik and Ravanbakhsh, Siamak and Poczos, Barnabas and Salakhutdinov, Russ R and Smola, Alexander J},
  journal={Advances in neural information processing systems},
  volume={30},
  year={2017}
}

@inproceedings{lee2019set,
  title={Set transformer: A framework for attention-based permutation-invariant neural networks},
  author={Lee, Juho and Lee, Yoonho and Kim, Jungtaek and Kosiorek, Adam and Choi, Seungjin and Teh, Yee Whye},
  booktitle={International conference on machine learning},
  pages={3744--3753},
  year={2019},
  organization={PMLR}
}

@article{kipf2016semi,
  title={Semi-supervised classification with graph convolutional networks},
  author={Kipf, Thomas N and Welling, Max},
  journal={arXiv preprint arXiv:1609.02907},
  year={2016}
}

@article{velivckovic2017graph,
  title={Graph attention networks},
  author={Veli{\v{c}}kovi{\'c}, Petar and Cucurull, Guillem and Casanova, Arantxa and Romero, Adriana and Lio, Pietro and Bengio, Yoshua},
  journal={arXiv preprint arXiv:1710.10903},
  year={2017}
}

@article{vaswani2017attention,
  title={Attention is all you need},
  author={Vaswani, Ashish and Shazeer, Noam and Parmar, Niki and Uszkoreit, Jakob and Jones, Llion and Gomez, Aidan N and Kaiser, {\L}ukasz and Polosukhin, Illia},
  journal={Advances in neural information processing systems},
  volume={30},
  year={2017}
}

@article{raissi2019physics,
  title={Physics-informed neural networks: A deep learning framework for solving forward and inverse problems involving nonlinear partial differential equations},
  author={Raissi, Maziar and Perdikaris, Paris and Karniadakis, George E},
  journal={Journal of Computational physics},
  volume={378},
  pages={686--707},
  year={2019},
  publisher={Elsevier}
}

@book{altman2021constrained,
  title={Constrained Markov decision processes},
  author={Altman, Eitan},
  year={2021},
  publisher={Routledge}
}

@misc{waissi1994network,
  title={Network flows: Theory, algorithms, and applications},
  author={Waissi, Gary R},
  year={1994},
  publisher={JSTOR}
}

@book{cormen2022introduction,
  title={Introduction to algorithms},
  author={Cormen, Thomas H and Leiserson, Charles E and Rivest, Ronald L and Stein, Clifford},
  year={2022},
  publisher={MIT press}
}

@inproceedings{tassiulas1990stability,
  title={Stability properties of constrained queueing systems and scheduling policies for maximum throughput in multihop radio networks},
  author={Tassiulas, Leandros and Ephremides, Anthony},
  booktitle={29th IEEE Conference on Decision and Control},
  pages={2130--2132},
  year={1990},
  organization={IEEE}
}

@book{georgiadis2006resource,
  title={Resource allocation and cross-layer control in wireless networks},
  author={Georgiadis, Leonidas and Neely, Michael J and Tassiulas, Leandros},
  year={2006},
  publisher={Now Publishers Inc}
}

@book{boyd2004convex,
  title={Convex optimization},
  author={Boyd, Stephen and Vandenberghe, Lieven},
  year={2004},
  publisher={Cambridge university press}
}

@article{schulman2017proximal,
  title={Proximal policy optimization algorithms},
  author={Schulman, John and Wolski, Filip and Dhariwal, Prafulla and Radford, Alec and Klimov, Oleg},
  journal={arXiv preprint arXiv:1707.06347},
  year={2017}
}

@inproceedings{burges2005learning,
  title={Learning to rank using gradient descent},
  author={Burges, Chris and Shaked, Tal and Renshaw, Erin and Lazier, Ari and Deeds, Matt and Hamilton, Nicole and Hullender, Greg},
  booktitle={Proceedings of the 22nd international conference on Machine learning},
  pages={89--96},
  year={2005}
}

@inproceedings{achiam2017constrained,
  title={Constrained policy optimization},
  author={Achiam, Joshua and Held, David and Tamar, Aviv and Abbeel, Pieter},
  booktitle={International conference on machine learning},
  pages={22--31},
  year={2017},
  organization={Pmlr}
}

@techreport{halpern2015service,
  title={Service function chaining (SFC) architecture},
  author={Halpern, Joel and Pignataro, Carlos},
  year={2015}
}

@article{foukas2017network,
  title={Network slicing in 5G: Survey and challenges},
  author={Foukas, Xenofon and Patounas, Georgios and Elmokashfi, Ahmed and Marina, Mahesh K},
  journal={IEEE communications magazine},
  volume={55},
  number={5},
  pages={94--100},
  year={2017},
  publisher={IEEE}
}

@article{ordonez2017network,
  title={Network slicing for 5G with SDN/NFV: Concepts, architectures, and challenges},
  author={Ordonez-Lucena, Jose and Ameigeiras, Pablo and Lopez, Diego and Ramos-Munoz, Juan J and Lorca, Javier and Folgueira, Jesus},
  journal={IEEE Communications Magazine},
  volume={55},
  number={5},
  pages={80--87},
  year={2017},
  publisher={IEEE}
}

@article{mao2017survey,
  title={A survey on mobile edge computing: The communication perspective},
  author={Mao, Yuyi and You, Changsheng and Zhang, Jun and Huang, Kaibin and Letaief, Khaled B},
  journal={IEEE communications surveys \& tutorials},
  volume={19},
  number={4},
  pages={2322--2358},
  year={2017},
  publisher={IEEE}
}

@article{kreutz2014software,
  title={Software-defined networking: A comprehensive survey},
  author={Kreutz, Diego and Ramos, Fernando MV and Verissimo, Paulo Esteves and Rothenberg, Christian Esteve and Azodolmolky, Siamak and Uhlig, Steve},
  journal={Proceedings of the IEEE},
  volume={103},
  number={1},
  pages={14--76},
  year={2014},
  publisher={Ieee}
}

@article{bari2012data,
  title={Data center network virtualization: A survey},
  author={Bari, Md Faizul and Boutaba, Raouf and Esteves, Rafael and Granville, Lisandro Zambenedetti and Podlesny, Maxim and Rabbani, Md Golam and Zhang, Qi and Zhani, Mohamed Faten},
  journal={IEEE communications surveys \& tutorials},
  volume={15},
  number={2},
  pages={909--928},
  year={2012},
  publisher={IEEE}
}

@article{tessler2018reward,
  title={Reward constrained policy optimization},
  author={Tessler, Chen and Mankowitz, Daniel J and Mannor, Shie},
  journal={arXiv preprint arXiv:1805.11074},
  year={2018}
}

@inproceedings{liu2020ipo,
  title={Ipo: Interior-point policy optimization under constraints},
  author={Liu, Yongshuai and Ding, Jiaxin and Liu, Xin},
  booktitle={Proceedings of the AAAI conference on artificial intelligence},
  volume={34},
  number={04},
  pages={4940--4947},
  year={2020}
}

\end{document}